\documentclass[11pt]{article}

\usepackage{times}
\usepackage{paralist}
\usepackage{amsmath}
\usepackage{amsthm}
\usepackage{latexsym}
\usepackage{url}
\usepackage[pdfpagelabels]{hyperref}
\usepackage{boxedminipage}
\usepackage{array}
\usepackage[ruled]{algorithm2e}
\usepackage{flafter}
\usepackage{placeins}

\usepackage{amssymb}
\usepackage{color}
\usepackage{makeidx}
\usepackage[font={small}]{caption}

\usepackage{tikz}
\usetikzlibrary{calc}
\usetikzlibrary{intersections}
\usetikzlibrary{shapes.geometric}

\usepackage{float}
\usepackage{graphicx}

\graphicspath{{graphics/}}

\usepackage[paper=letterpaper,margin=1in]{geometry}

\usepackage{soul}

\usepackage{etoolbox}

\newbool{ext}

\newbool{shortVersion}

\ifbool{shortVersion}{
\usepackage[nomarkers,figuresonly,nofiglist]{endfloat}
}{}

\newtheorem{theorem}{Theorem}
\newtheorem*{theorem*}{Theorem}

\newtheorem*{conjecture*}{Conjecture}

\newtheorem*{remark*}{Remark}
\newtheorem{lemma}[theorem]{Lemma}
\newtheorem*{lemma*}{Lemma}
\newtheorem{definition}[theorem]{Definition}
\newtheorem*{definition*}{Definition}
\newtheorem{proposition}[theorem]{Proposition}
\newtheorem*{proposition*}{Proposition}
\newtheorem{corollary}[theorem]{Corollary}
\newtheorem*{corollary*}{Corollary}

\newtheorem*{property*}{Property}

\newtheorem*{lem2shen*}{Lemma 2 in~\cite{shen_1999_finding}}

\DeclareMathOperator{\TSP}{TSP}
\DeclareMathOperator{\MST}{MST}

\DeclareMathOperator{\conv}{conv}

\DeclareMathOperator{\OPT}{OPT}

\DeclareMathOperator{\val}{val}
\DeclareMathOperator{\argmax}{argmax}

\def\b1{{\bf 1}}

\title{An $O(1)$-Approximation for Minimum
Spanning Tree Interdiction}

\author{
Rico Zenklusen%
\thanks{
Department of Mathematics, ETH Zurich, Zurich, Switzerland,
and Department of Applied Mathematics and Statistics,
Johns Hopkins University, Baltimore, USA.
Email:
\href{mailto:ricoz@math.ethz.ch}{\tt ricoz@math.ethz.ch}.
}
}

\begin{document}

\maketitle

\begin{abstract}
Network interdiction problems are a natural way
to study the sensitivity of a network optimization
problem with respect to the removal of a limited
set of edges or vertices.
One of the oldest and best-studied interdiction
problems is minimum spanning tree (MST) interdiction.
Here, an undirected multigraph
with nonnegative edge weights
and positive interdiction costs on its edges is
given, together with a positive budget $B$.
The goal is to find a subset of edges $R$, whose total
interdiction cost does not exceed $B$, such that
removing $R$ leads to a graph where the
weight of an MST is as large as possible.
Frederickson and Solis-Oba (SODA 1996) presented
an $O(\log m)$-approximation for MST interdiction,
where $m$ is the number of edges. Since then,
no further progress has been made regarding
approximations, 
and the question 
whether MST interdiction admits an $O(1)$-approximation
remained open.

We answer this question 
in the affirmative, by presenting
a $14$-approximation that overcomes two main
hurdles that hindered further progress so far.
Moreover, based on a well-known $2$-approximation for
the metric traveling salesman problem (TSP),
we show that our $O(1)$-approximation for MST
interdiction implies an $O(1)$-approximation for
a natural interdiction version of metric TSP.

\end{abstract}


\thispagestyle{empty}

\newpage

\setcounter{page}{1}

\section{Introduction}\label{sec:intro}

Network interdiction studies the sensitivity of a network
optimization problem with respect to the removal of some limited
set of its edges or vertices. 
For example, in the minimum spanning tree (MST) interdiction
problem,
we are given an undirected loopless
multigraph $G=(V,E)$ with nonnegative
edge weights $w:E\rightarrow \mathbb{Z}_{\geq 0}$, positive
edge interdiction costs $c:E\rightarrow \mathbb{Z}_{>0}$, and an
interdiction budget $B\in \mathbb{Z}_{>0}$. The goal is to remove
a subset of edges whose total interdiction cost is bounded by
$B$, and such that the weight of an MST in the graph on the
non-removed edges is as large as possible.
To avoid trivial cases, we assume that the budget is not
large enough to disconnect the graph.
Along the same lines, interdiction problems have
been considered for a wide variety of other underlying network
optimization problems,
including maximum $s$-$t$ flows, maximum matchings, shortest
paths, maximum edge-connectivity, and
maximum stable sets (see Section~\ref{subsec:relatedWork}
for references and some further details).
As highlighted in the example of interdicting MSTs,
interdiction problems can naturally be interpreted as
two-player problems, where an \emph{interdictor}
first removes edges and plays against an \emph{operator},
who solves an optimization problem over the remaining network.

Interdiction problems allow for identifying weak spots in
a networked system that may be worth reinforcing, or to
obtain strategies to interdict an optimization problem
that describes an undesirable process on a network%
.
Therefore, interdiction problems
have found applications in a wide variety
of areas, including preventing the spread of infections
in hospitals~\cite{assimakopoulos_1987_network},
inhibiting the production and distribution of illegal
drugs~\cite{wood_1993_deterministic},
prevention of nuclear arms smuggling~\cite{morton_2007_models},
military planning~\cite{ghare_1971_optimal}, 
and infrastructure
protection~\cite{salmeron_2009_worst,church_2004_identifying}.
Even the discovery of the Max-Flow/Min-Cut Theorem
was motivated by a Cold War plan to interdict the Soviet
rail network in Eastern Europe~\cite{schrijver_2002_history}.

Considerable effort has also been spent in getting a better
theoretical understanding of interdiction problems. However,
large gaps remain. This is especially true
regarding their approximability,
which is of particular interest since almost all known interdiction
problems are easily shown to be NP-hard.
One of the oldest and most-studied interdiction problems,
for which a large gap in terms of approximability exists,
is MST interdiction, which is the focus of this paper.
It captures well-known graph optimization problems,
like the
\emph{maximum components problem}
(MCP)~\cite{frederickson_1999_increasingC},
which asks to break a graph into as many connected
components
as possible by removing a given number $q$ of edges.
Also the generalization of MCP with interdiction
costs on the edges and an interdiction budget $B$, which
was studied in~\cite{engelberg_2007_cut} and called 
the \emph{budgeted graph disconnection} (BGD)
problem, remains
a special case of MST interdiction.
Notice the close relation between MCP
and the $k$-cut problem~\cite{goldschmidt_1994_polynomial},
where the roles of objective and budget are exchanged.
In particular, as observed
in~\cite{frederickson_1999_increasingC},
this connection to the $k$-cut problem
immediately implies strong NP-hardness of MCP,
and therefore also of MST interdiction.
\ifbool{shortVersion}{%
\footnote{%
}%
}{%
For completeness,
we briefly discuss this connection
in Appendix~\ref{app:graphPart}.
}
Another motivation for studying MST interdiction is that
MSTs are often used as building blocks in
other optimization problems or approximation algorithms.
Results on MST interdiction 
therefore have the potential to be carried over to
further interesting problem settings.
In particular, we
exploit the well-known property that the weight of
an MST is within a factor of $2$ of the shortest tour
for the metric traveling
salesman problem (TSP), to transform approximation
results on MST interdiction to metric TSP interdiction.

In 1996, Frederickson and
Solis-Oba~\cite{frederickson_1999_increasingC} presented an
$O(\log m)$-approximation for MST interdiction,
where $m=|E|$ is the number of edges.
No improvement on the approximation ratio has been
obtained since.
We highlight that parallel edges are allowed in
the MST interdiction problem, and we
thus may have $\log m = \omega(\log n)$, where $n=|V|$.
Admitting parallel edges is of particular interest
in MST interdiction and also other interdiction problems,
since they allow for modeling effects like partial
destruction of a connection between two vertices.
Hence, so far, no approximation algorithm for MST
interdiction is known with an approximation factor
that is polylogarithmic in $n$.

A special case of MST interdiction, which received
considerably attention, is
the \emph{$k$ most vital edges problem}, which
asks to remove $k$ edges to obtain a graph whose
MST has a weight as large as possible. Hence, this
corresponds to MST interdiction with unit interdiction
costs and budget $B=k$.
From an approximation point of view,
the best known procedure is as well the algorithm
of Frederickson and Solis-Oba. However, 
for the $k$ most vital edges problem this algorithm
is known to be an
$O(\log k)$-approximation~\cite{frederickson_1999_increasingC}.
Interest arose in obtaining fast polynomial algorithms
for $k=O(1)$.
In particular, the \emph{most vital edge problem},
which corresponds to $k=1$, is
closely related to the \emph{sensitivity analysis problem}
for MSTs, as observed in~\cite{iwano_1993_efficient}. 
In the sensitivity analysis problem one is given an
edge-weighted graph $G=(V,E)$ and an MST $T\subseteq E$
in $G$. The goal
is to determine for every edge by how much its weight
has to be changed so that $T$ is not anymore an MST.
Clearly, any algorithm to find an MST combined with
an algorithm for the sensitivity analysis problem
leads to an algorithm to solve the most vital edge
problem.
Using this observation leads to the currently fastest
algorithms for the most vital edge problem, beating
the strongest specialized approaches known
previously~\cite{hsu_1991_finding}.
In particular, a deterministic
$O(m \cdot \alpha(m,n))$ time algorithm for the most
vital edge problem is obtained---where
$\alpha(m,n)$ is the inverse Ackermann function---by
combining Chazelle's~\cite{chazelle_2000_minimum}
$O(m \cdot \alpha(m,n))$ MST
algorithm with Tarjan's~\cite{tarjan_1979_applications}
$O(m \cdot \alpha(m,n))$ algorithm for the sensitivity analysis
problem.
Moreover, a randomized $O(m)$ time algorithm is obtained
for the most vital edge problem by
combining an $O(m)$ randomized MST algorithm---like
the original algorithm of Klein and Tarjan~\cite{klein_1994_randomized}
or a revised version presented by
Karger et
al.~\cite{karger_1995_randomizedLinear-Time}---with
a randomized $O(m)$ time algorithm by Dixon et
al.~\cite{dixon_1992_verification}
for the sensitivity analysis problem.\footnote{We highlight
that a simpler randomized $O(m)$ time algorithm for
the sensitivity analysis problem was later obtained
by King~\cite{king_1997_simpler}.}
Pettie~\cite{pettie_2005_sensitivity}
presented an even faster
deterministic $O(m\cdot \log \alpha(m,n))$ time
algorithm for the sensitivity analysis problem. However, this
does not lead to improvements for currently fastest deterministic
algorithms for the most vital edge problem because no
deterministic method is known to find an MST faster than in
$O(m\cdot \alpha(m,n))$ time.
Several exponential-time algorithms have been suggested
for the $k$ most vital edges problem for
general
$k$, including parallel
algorithms~\cite{liang_1997_finding,liang_2001_finding,bazgan_2012_efficient}.
The problem has also been considered under the aspect
of parameterized complexity~\cite{guo_2014_parameterized}.

Our focus on MST interdiction lies on approximation algorithms.
From an approximation point of view, the central open
question within MST interdiction is whether it is possible
to obtain an $O(1)$-approximation.
The main contribution of this paper is to answer this
question in the affirmative.
As a direct consequence thereof, we obtain
an $O(1)$-approximation for a natural interdiction version
of metric TSP.

\subsection{Our results and techniques}

Our main result is the first $O(1)$-approximation
for MST interdiction, improving on Frederickson
and Solis-Oba's
$O(\log m)$-approximation~\cite{frederickson_1999_increasingC}.

\begin{theorem}\label{thm:MSTmain}
There is a $14$-approximation for MST interdiction.
\end{theorem}

MSTs are a useful tool in approximation algorithms for
other combinatorial optimization problems, like metric
TSP. Due to this link, we can use the above result
as a black-box to obtain
an $O(1)$-approximation for a natural interdiction version
of metric TSP.
In metric TSP, a complete graph is given
with lengths on the edges that 
satisfy the triangle inequality, and the task is
to find a shortest Hamiltonian cycle.
Metric TSP often stems from 
settings where a graph $G=(V,E)$ with
edge lengths $\ell:E\rightarrow \mathbb{Z}_{>0}$ is given,
and the goal is to find a shortest closed walk that
visits every vertex \emph{at least} once.
Such settings easily translate to metric TSP 
by considering a complete graph $\overline{G}=(V,\overline{E})$
over $V$
such that to every edge $\{u,v\}\in\overline{E}$
the distance $d(\{u,v\})$
is assigned, where $d(\{u,v\})$ is the length of a shortest
$u$-$v$ path in $G$.
A natural interdiction version is obtained by considering
interdiction costs $c:E\rightarrow \mathbb{Z}_{>0}$ in $G$
and a budget $B\in \mathbb{Z}_{>0}$;
the task is to find a subset of edges
$R\subseteq E$ such that the shortest closed walk in
$(V,E\setminus R)$ that visits each vertex at least once
is as large as possible. For brevity, we call this
problem \emph{metric TSP interdiction}.
Combining Theorem~\ref{thm:MSTmain} with a well-known
$2$-approximation for metric TSP that
is based on MSTs, we obtain:

\begin{theorem}\label{thm:TSPmain}
Metric TSP interdiction admits a $28$-approximation.
\end{theorem}

To obtain our main result, Theorem~\ref{thm:MSTmain},
we overcome two main hurdles
for obtaining $O(1)$-approximations
for MST interdiction.
First, it is hard to find a good upper bound for
MST interdiction. In particular, no strong LP
relaxations are known.
We note that even for the related $k$-cut problem and
variants of it, it is nontrivial to find
LP relaxations with constant integrality
gap (see~\cite{naor_2001_tree,%
chekuri_2006_steiner,%
ravi_2008_approximating}
and references therein).

A second obstacle, which also makes clear why MST
interdiction seems substantially more difficult
to approximate than MCP, 
is the fact that MST interdiction can be
interpreted as a multilevel BGD problem, with
interactions between the levels that are hard
to control.
To highlight this connection, which goes back
to~\cite{frederickson_1999_increasingC},
we first observe that one can assume that each edge
weight is either zero or a power of two, by losing
at most a factor of $2$ in the approximation guarantee.
This is achieved by rounding down all edge weights
to the next power of $2$ (without changing zero-edges).
Let $E_{\leq i}$ be all edges with weight
at most $2^i$. 
Now one can observe, and we will formalize this in
Section~\ref{sec:preliminaries}, that the weight
of an MST is determined by the number of connected
components of $G_i=(V,E_{\leq i})$ for each $i$.
Hence, MST interdiction seeks to break the graphs
$G_i$ into as many components as possible, 
where breaking a graph $G_i$ into an additional component
has an impact on the weight of MSTs that is the higher,
the larger the index $i$ is.
The approximation algorithm of
Frederickson and Solis-Oba~\cite{frederickson_1999_increasingC}
essentially focusses only on one level where a high 
impact can be achieved, thus reducing the problem to 
a BGD problem, or an MCP 
for the case of unit interdiction costs.
No algorithm is known so far that exploits the interactions
between the different levels, which seems crucial for
obtaining $O(1)$-approximations.

The way we address these two obstacles is as follows.
First we obtain a good upper bound $\nu^*$ for the optimal
value $\OPT$ by formulating a parametric submodular
minimization problem.
However, instead of finding a way to directly compare
against $\nu^*$, we focus on what we call \emph{efficiencies}
of potential edge sets to remove.
More precisely, the efficiency
of a set $U\subseteq E$---which does not need to fulfill
the budget constraint---is defined as follows.
Let $\val(U)$ be the weight of an MST in $(V,E\setminus U)$.
Then the efficiency of $U$ is given by $\val(U)/c(U)$.
Apart from simple special cases, our algorithm computes
a set $U\subseteq E$ that is over budget, and whose efficiency
is close to $\nu^*/B$, which is at least as good as
the efficiency of an optimal interdiction set.
The core part of our algorithm is a procedure that,
given a set $U\subseteq E$ with $c(U) > B$, computes a
set $R\subseteq U$ fulfilling the budget constraint and
whose efficiency is close to the efficiency of $U$.
Since we choose $U$ to have a close-to-optimal efficiency,
this allows us to compare our solution to $\nu^*$.

To design this core part of the algorithm, we overcome
the above-explained difficulty coming from the interpretation
of MST interdiction as multilevel BGD problem as follows.
We exploit that $U\subseteq E$ is a
high-efficiency set, which implies that it has a good
overall impact over the different levels $i$.
To obtain a solution $R$ that largely inherits this property
from $U$, we start with $R=\emptyset$ and successively add to
$R$ appropriate subsets of $U$ that are guaranteed to have
a good impact over several levels, as long as $c(R)\leq B$.

We highlight that, in the interest of clarity,
we did not try to heavily optimize constants.

\subsection{Further related work}\label{subsec:relatedWork}

Many interdiction problems beyond the minimum spanning tree 
setting have been studied.
This includes interdiction versions
of the maximum $s$-$t$ flow
problem~%
\cite{wood_1993_deterministic,phillips_1993_network,%
zenklusen_2010_network}
(a setting often called \emph{network flow interdiction}),
the shortest path problem~\cite{ball_1989_finding,khachiyan_2008_short},
the maximum matching
problem~\cite{zenklusen_2010_matching,dinitz_2013_packing},
interdicting the connectivity of
a graph~\cite{zenklusen_2014_connectivity},
interdiction of packings~\cite{dinitz_2013_packing},
stable set interdiction~\cite{bazgan_2011_most},
and variants of facility
location~\cite{church_2004_identifying}. 
However, the theoretical understanding
of most interdiction problems still seems rather limited.
A good example for which a large gap remains between the best
known hardness results and approximation algorithms is network
flow interdiction. Network flow interdiction is a strongly NP-hard
problem~\cite{wood_1993_deterministic} for which no 
approximation results are known, except for a
pseudo-approximation~\cite{burch_2002_decompositionbased}
which is allowed to violate the budget by a factor of $2$.

A related line of research is the study of a \emph{continuous}
version of interdiction problems, where the weight
of edges can be increased
continuously at a given weight per cost ratio which depends on
the edge.
These models are typically much more tractable then their
discrete counterparts, i.e., the classical interdiction problems.
The reason for this is that they can often be written as
a single linear program. In particular, efficient
algorithms for continuous
interdiction have been obtained for maximum weight independent
set in a matroid~\cite{frederickson_1999_increasingC},
maximum weight common independent sets in two
matroids and the minimum
cost circulation problem~\cite{juettner_2006_budgeted}.

We highlight
that~\cite{shen_1999_finding} claims to
present a $2$-approximation for
the $k$ most vital edges problem for MST.
However, the results in~\cite{shen_1999_finding}
are based on an erroneous lemma about spanning trees.
\ifbool{shortVersion}{
We provide details on this erroneous lemma in
the long version of the paper.
}{
In Appendix~\ref{app:shen} we provide details on this
erroneous lemma.
}

\subsection*{Organization of the paper}

We formally define the problem and present some
basic observations in Section~\ref{sec:preliminaries}.
Section~\ref{sec:highLevel} outlines our algorithmic
approach, and reduces the task of finding an
$O(1)$-approximation for MST interdiction to one
specific subproblem, for which we present
an algorithm in Section~\ref{sec:mainAlg}.
The analysis of this algorithm is provided in
Section~\ref{sec:analysis}.
\ifbool{shortVersion}{
All missing proofs can be found in the long version
of the paper. All figures have been placed at the
end of the document.
}{
Finally,
Section~\ref{sec:TSPInt} provides the details of
our result for metric TSP interdiction, thus proving
Theorem~\ref{thm:TSPmain}.}

\section{Preliminaries}\label{sec:preliminaries}

Throughout this paper, $G=(V,E)$ is an undirected multigraph
with edge weights $w:E\rightarrow \mathbb{Z}_{\geq 0}$,
edge costs $c:E\rightarrow \mathbb{Z}_{>0}$, and
a global budget $B\in \mathbb{Z}_{>0}$.
Furthermore, we assume that each edge weight
is either a power of two or zero, i.e.,
$w:E\rightarrow \{0,1,\dots, 2^p\}$. This can be achieved by
rounding down all edge weights to the next power of two
(without changing zero-edges). Clearly, this rounding
changes the weight of any MST in $G$ or any of its
subgraphs by at most a factor of two.
Hence, any $\alpha$-approximation for MST interdiction
with weights being powers of two is a $2\alpha$-approximation
for general MST interdiction.

The MST interdiction problem asks to find a subset
of edges $R\subseteq E$ with
$c(R)\leq B$ that maximizes the weight of an MST 
in $(V, E\setminus R)$; we denote the weight of such
an MST by $\val(R)$.
Hence, $\val(\emptyset)$ is the weight of a minimum spanning
tree in $G$, and the MST interdiction problem
can formally be described by
\begin{equation}\label{eq:mstIntProb}
\max \{\val(R) \;\vert\; R\subseteq E, c(R) \leq B\}.
\end{equation}
Let $\OPT$ be the optimal value of
problem~\eqref{eq:mstIntProb}.
We call a set $R\subseteq E$ with $c(R)\leq B$
an \emph{interdiction set}.
When talking about edge sets $U\subseteq E$
that may not satisfy the budget constraint,
but about which we still think of edges to
be removed, we use the notion \emph{removal set}.

To easily distinguish the different weight-levels we
define 
\begin{gather*}
E_{-1} = \{e\in E \mid w(e) = 0\}, \quad
E_i    = \{e\in E\mid w(e) = 2^i\}
  \quad \forall i\in \{0,\dots, p\}, \text{ and}\\
E_{\leq i} = E_{-1} \cup \dots \cup E_i
  \quad \forall i\in \{-1,\dots, p\}.
\end{gather*}
To avoid trivial cases, we assume that 
no interdiction set disconnects
the graph, i.e.,
$c(\delta(S)) > B$ for all $S\subsetneq V, S\neq \emptyset$,
where $\delta(S)\subseteq E$ is the set of all edges with
precisely one endpoint in $S$.
Due to this, there is always an optimal interdiction set
that does not remove any edge from $E_p$, i.e., the
edges with heaviest weight.
Indeed, removing edges of heaviest weight
cannot increase the weight of an MST, except
if one could break the graph into several components,
which we excluded.
For simplicity we can therefore assume that
$(V,E_p)$ is a connected graph.
This can be achieved
by adding a non-interdictable spanning tree consisting of
edges of weight $2^p$ to $G$.
By the above discussion, adding such edges does
not have any impact on the MST interdiction
problem. Since there are optimal interdiction sets
not containing any edge of $E_p$, we will consider
throughout the paper only removal sets that are
subsets of $E_{\leq p-1}$.
Moreover, we assume to have at least $3$
levels, i.e., $p\geq 1$, to simplify the presentation.

Furthermore, we assume that there is an
interdiction set $R\subseteq E_{\leq p-1}$
such that $(V,E_{\leq p-1}\setminus R)$ has more connected components
than $(V,E_{\leq p-1})$. Without this assumption,
there is no interdiction set $R$ that increases the number
of edges in $E_p$ that must be used
in any MST in $(V,E\setminus R)$.
In such a case, independent of the interdiction set $R$,
any MST in $(V, E\setminus R)$ would use the same
number of edges in $E_p$, namely a minimal
set of edges connecting the connected components 
of $(V,E_{\leq p-1})$.
Hence, one could reduce the
problem by contracting any minimum edge set 
in $E_p$ that connects the connected
components of $(V,E_{\leq p-1})$.

For our analysis we focus on a well-known formula
to describe the weight of an MST, which highlights the
level-structure.
For $U\subseteq E$, let $\sigma(U)$ be the number of
connected components of the graph $(V,U)$.
For any $U\subseteq E_{\leq p-1}$, the weight $\val(U)$ of
an MST in $(V,E\setminus U)$ is given by
\begin{equation}\label{eq:valLevels}
\val(U) = \sigma(E_{-1} \setminus U) - 1
  + \sum_{i=0}^{p-1} 2^{i}
  \bigg(\sigma(E_{\leq i} \setminus U) - 1\bigg).
\end{equation}
This formula readily follows from the optimality
of the greedy algorithm to find an MST, or from 
known results on matroid
optimization~(see \cite[Volume B]{schrijver_2003_combinatorial}).%
\footnote{In particular,~\eqref{eq:valLevels} is a consequence
of Theorem~40.2 in~\cite{schrijver_2003_combinatorial},
which describes the weight of a maximum spanning tree in terms
of the rank function 
$r:2^E\rightarrow \mathbb{Z}_{\geq 0}$ of the
graphic matroid, which satisfies 
$r(U) = n-\sigma(U)$.
Notice that the MST problem can easily
be reduced to the maximum spanning tree problem
with nonnegative weights
by replacing each edge weight $w(e)$ by
$M-w(e)$ for a sufficiently large constant $M$.
}
Furthermore, it shows explicitly that for every
additional component that is created on level
$i\in \{0,\dots, p-1\}$---i.e., in the graph
$(V,E_{\leq i})$---when removing $U$, the
weight of MSTs increases by $2^i$.
Moreover, we highlight the well-known fact that $\sigma(U)$, and
therefore also $\sigma(E_{\leq i}\setminus U)$ for
$i\in \{-1,\dots, p-1\}$, is a supermodular function
in $U$, i.e.,
$\sigma(A) + \sigma(B) \leq \sigma(A\cup B) + \sigma(A\cap B)$
for $A,B\subseteq E$. This follows from the fact
that $\sigma(U) = n - r(U)$, where $r$ is
the rank function of the graphic matroid, which is
submodular.
This also implies that $\val(U)$ is supermodular in $U$,
a fact we use later to find a removal set of high efficiency
via submodular function minimization.

\section{Outline of our approach}\label{sec:highLevel}

A core part of our algorithm is described in the following
theorem.
Before proving the theorem in Section~\ref{sec:mainAlg},
we will show 
how it can be used to obtain an $O(1)$-approximation
for MST interdiction.

\begin{theorem}\label{thm:mainAlgThm}
There is an efficient algorithm (to be described in
Section~\ref{sec:mainAlg})
that, for any set $U\subseteq E_{\leq p-1}$ with 
$c(U)>B$,
returns an interdiction set $R\subseteq E$ with
\begin{equation}\label{eq:mainAlgThm}
\val(R) \geq  \frac{1}{2} \cdot B\cdot
  \frac{\val(U)}{c(U)} - 2^{p+1}.
\end{equation}
\end{theorem}

We can get rid of the additive term $2^{p+1}$
in~\eqref{eq:mainAlgThm} by a best-of-two
algorithm that either returns the interdiction set $R$
claimed by Theorem~\ref{thm:mainAlgThm} or
an interdiction set that increases the number
of connected components in $(V,E_{\leq p-1})$,
which exists by assumption.

\begin{corollary}\label{cor:lowBoundFromU}
There is an efficient algorithm 
that, for any set $U\subseteq E$ with $c(U)>B$,
returns an interdiction set $R\subseteq E$ with
\begin{equation*}
\val(R) \geq \frac{1}{6} \cdot B\cdot
  \frac{\val(U)}{c(U)}.
\end{equation*}
\end{corollary}
\ifbool{shortVersion}{}{%
\begin{proof}
Let $U\subseteq E$ with $c(U)>B$,
and let $R_1\subseteq E$ be an interdiction set
as claimed by Theorem~\ref{thm:mainAlgThm}.
Furthermore, by assumption,
there exists an interdiction set $R_2\subseteq E$, such that
$(V,E_{\leq p-1} \setminus R_2)$ has at least two
components.
(Actually, the assumption even implies that there is
an interdiction set $R_2$ such that
$(V,E_{\leq p-1} \setminus R_2)$ has at least one
more component than $(V,E_{\leq p-1})$.)
Such a set $R_2$ can be found efficiently
by finding a minimum cost cut in $(V,E_{\leq p-1})$.
Hence, $\val(R_2)\geq 2^p$.
Let $R\in \argmax_{i\in \{1,2\}}\val(R_i)$.
The set $R$ satisfies the conditions of
Theorem~\ref{cor:lowBoundFromU} since
\begin{align*}
\frac{1}{2} B \frac{\val(U)}{c(U)}
  &\leq \val(R_1) + 2^{p+1}
    && \text{(by~\eqref{eq:mainAlgThm})}\\
  &\leq \val(R_1) + 2 \val(R_2)\\
  &\leq 3 \val(R).
\end{align*}

\end{proof}
}

In the following we show that either
we can get an $O(1)$-approximation to MST interdiction
with a quite direct approach, or we can find a removal
set $U\subseteq E$ with $c(U)>B$ and high
efficiency $\val(U)/c(U)$.
For this, we take a somewhat different, bi-objective
look on MST interdiction that is independent 
of the budget value $B$.
Namely, for all sets $U\subseteq E_{\leq p-1}$ we
consider the tuple $(c(U), \val(U))$.
We are interested in sets $U\subseteq E$ with
a large MST value $\val(U)$ and small cost $c(U)$,
which can be interpreted as two objectives on $U$.
Using standard notions of multi-objective optimization,
we say that a tuple 
$(c(U), \val(U))$ is \emph{non-dominated} if
there is no other set $U'\subseteq E_{\leq p-1}$
with $c(U')\leq c(U)$, $\val(U')\geq \val(U)$ and at
least one of these two inequalities being strict.
The Pareto front in the cost-value space consists
therefore of
all non-dominated tuples $(c(U),\val(U))$
for $U\subseteq E_{\leq p-1}$,
which can all be interpreted 
as optimal solutions to problem~\eqref{eq:mstIntProb}
when varying the budget.

Whereas finding a particular point on the Pareto
front through solving problem~\eqref{eq:mstIntProb}
is NP-hard (since it is precisely the MST interdiction
problem), one can efficiently compute so-called
\emph{extreme supported solutions} or
\emph{extreme supported tuples},
which are 
all vertices of
$\conv(\{ (c(U),\val(U)) \mid U\subseteq E_{\leq p-1}\})
+\mathbb{R}_{\geq 0} \times \mathbb{R}_{\leq 0}$,
where $\conv$ is the convex hull operator.
Hence, a tuple $(c(U), \val(U))$ for some
$U\subseteq E_{\leq p-1}$ is an extreme supported tuple
if there is a $\lambda \geq 0$ such that
this tuple is the unique minimizing tuple for
\begin{equation}\label{eq:attackProb}
\min\{ \lambda\cdot c(U) - \val(U)
  \mid U\subseteq E_{\leq p-1} \}.\enspace%
\footnote{Notice that~\eqref{eq:attackProb} can also
be interpreted as a Lagrangean dual of
$\min\{-\val(R) \mid R\subseteq E, c(R) \leq B\}$.
We focus on the Pareto front interpretation since
it is natural for properties we want
to highlight later. The Lagrangian dual approach
has been employed in similar problems
in budgeted optimization 
(see~\cite{ravi_1996_constrained,grandoni_2014_new}
and references therein).}
\end{equation}
Notice that there may be several edge sets
$U\subseteq E_{\leq p-1}$ that correspond to the
same (extreme supported) tuple.
Figure~\ref{fig:case3}
\ifbool{shortVersion}{on page~\pageref{fig:case3}}{}
shows an example
of a Pareto front where the filled dots
correspond to all extreme supported
tuples, which we denote by $\mathcal{X}$.
Notice that for any $\lambda\geq 0$, the
objective $\lambda\cdot c(U) - \val(U)$ is a
submodular function in $U$ because $\val(U)$
is supermodular and $\lambda \cdot c(U)$ is
modular in $U$.
Problem~\eqref{eq:attackProb} is therefore
a parametric submodular function minimization
problem, which is a well-studied problem 
(see~\cite{fleischer_2003_push-relabel,%
nagano_2007_faster}).
In particular, there is a set of at most
$\beta\leq |E_{\leq p-1}|+1 \leq m+1$
different solutions
$U_1, \dots, U_\beta$, such that for each
$\lambda \geq 0$, one of these solutions is
optimal for~\eqref{eq:attackProb}.
In other words, the optimal value
of~\eqref{eq:attackProb} is a piecewise linear
function in $\lambda$ with at most $m+1$ segments.
The upper bound of $|E_{\leq p-1}|+1$ on $\beta$ follows by the
fact that one can choose sets $U_i$
that are nested.
Furthermore, Nagano~\cite{nagano_2007_faster}
showed that such a family of sets
$U_1,\dots, U_\beta$ can be determined
by a variation of Orlin's submodular function
minimization algorithm~\cite{orlin_2009_faster}
within the same strongly polynomial time complexity.
In summary, we can find in strongly polynomial
time all $O(m)$ points in $\mathcal{X}$ each 
with a corresponding set $U\subseteq E_{\leq p-1}$.

To find a good interdiction set, we distinguish
the following three cases, depending on the budget $B$.
\begin{enumerate}
\setlength{\itemsep}{0em}
\item[Case 1:] There is a tuple
$(\val(U),c(U))\in \mathcal{X}$ such that
$c(U)=B$. In this case $U$ is an optimal
solution to~\eqref{eq:mstIntProb} that we can
find efficiently and return. 

\item[Case 2:] $B$ is larger than the largest
first coordinate among all points in $\mathcal{X}$.
This implies that all edges in $E_{\leq p-1}$
can be removed simultaneously without exceeding
the budget. Hence, we return the interdiction
set $R=E_{\leq p-1}$ which is clearly optimal.

\item[Case 3:] There are two 
tuples $p_1=(c(U_1), \val(U_1)), p_2 = (c(U_2), \val(U_2))\in \mathcal{X}$
such that $c(U_1) < B < c(U_2)$, and $p_1$ and $p_2$
are consecutive in the sense that
there is no other tuple $(c(U),\val(U))\in \mathcal{X}$
with $c(U_1) < c(U) < c(U_2)$.
\end{enumerate}

Since we easily get optimal solutions for the first two cases,
we assume from now on to be in the third case.
Figure~\ref{fig:case3}
highlights a possible set $\mathcal{X}$ 
that corresponds to the third case.
\begin{figure}
\begin{center}
\begin{tikzpicture}[scale=1.2]

\begin{scope}[-stealth]
\draw (0,0) -- (6.5,0) node[right] {$c(U)$};
\draw (0,0) -- (0,5.5) node[above] {$\val(U)$};
\end{scope}

\begin{scope}[every node/.style={circle,
fill=black, draw=black,
outer sep=0pt, inner sep=0pt,
minimum size=4pt}]
\node (p1) at (0,0) {};
\node (p2) at (0.3,2) {};
\node (p3) at (1.2,3.5) {};
\node (p4) at (3.5,5) {};
\node (p5) at (5.5,5.3) {};
\node[minimum size=0pt] (p6) at (6.5,5.3) {};
\end{scope}

\begin{scope}[every node/.style={circle,
draw=black, fill=white, outer sep=0pt,
inner sep=0pt, minimum size=4pt}]

\node at (0.2,0.7) {};

\node at (0.7,2.2) {};
\node at (1.0,2.4) {};
\node at (1.1,2.9) {};

\node at (1.8,3.7) {};

\node (opt) at (2.3,3.9) {};

\node at (2.8,4.4) {};
\node at (3.3,4.6) {};
\node at (4.8,5.07) {};
\node at (5.2,5.14) {};
\end{scope}

\coordinate (ot) at ($(opt) + (-0.5,-1)$);
\draw[-stealth, shorten >= 2pt] (ot) -- (opt);
\node[below, align=center] at (ot) {optimal\\ solution};

\begin{scope}
\draw (p1) -- (p2) --
 (p3) -- (p4) -- (p5) -- (p6);
\end{scope}

\path[name path=p3p4] (p3) -- (p4);

\begin{scope}
\draw[name path=budget] (2.5,0) -- (2.5,5.5);
\node[below] at (2.5,0) {$B$};

\path [name intersections={of=p3p4
  and budget,by=int}];

\end{scope}

\coordinate (tf) at ($(int)+(-0.5,0.5)$);
\draw[-stealth, shorten >= 2pt] (tf) -- (int);

\begin{scope} 

\node[above left, xshift=14pt] at (p3)
  {$\begin{pmatrix}
c(U_1) \\
\val(U_1)
\end{pmatrix}$};

\node[above left, yshift=-3pt, xshift=5pt] at (tf)
  {$p=\begin{pmatrix}
B   \\
\nu^*
\end{pmatrix}$};

\node[above] at (p4)
  {$\begin{pmatrix}
c(U_2)\\
\val(U_2)
\end{pmatrix}$};

\end{scope}

\end{tikzpicture}

\end{center}
\caption{A possible constellation for
the third case.
The dots correspond to all non-dominated
solutions, i.e., to the Pareto front.
The filled dots represent the set
$\mathcal{X}$ of all extreme supported
tuples.
In the above example, the optimal tuple is not
part of $\mathcal{X}$.
}
\label{fig:case3}
\end{figure}
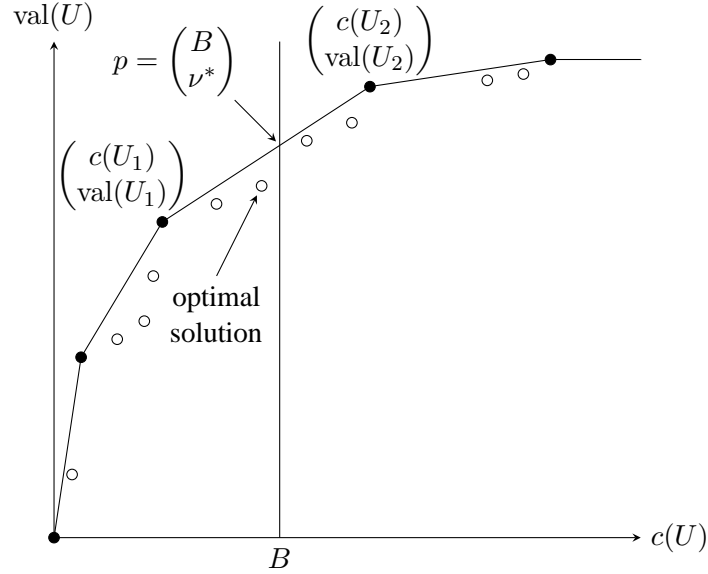
We can now upper bound $\OPT$ as follows.
Consider the point $p=(x,y)$ on the 
segment between $p_1$ and $p_2$  such that $x=B$
(see Figure~\ref{fig:case3}).
Clearly, $y$ is then equal to the following
value, which we denote by $\nu^*$:
\begin{equation}\label{eq:defFStar}
y=\nu^* = \val(U_1)
  + (B-c(U_1)) \frac{\val(U_2)-\val(U_1)}{c(U_2)-c(U_1)},
\end{equation}
and we have $\nu^*\geq \OPT$ 
since all solutions
are below the line that goes through $p_1$ and $p_2$,
because $p_1$ and $p_2$ are consecutive points on
the convex hull of the Pareto front.
We will show that the following algorithm
is an $O(1)$-approximation for the third
case.

\begin{algorithm} 
\DontPrintSemicolon
\eIf{$\val(U_1) \geq \frac{1}{7}\cdot  \nu^*$}{
\medskip
Return $U_1$.
}{
Return an interdiction  set
$R\subseteq E$ satisfying
\begin{equation*}
\val(R) \geq \frac{1}{6}\cdot
  B \cdot \frac{\val(U_2)}{c(U_2)}\quad, 
\end{equation*}
which can be obtained
by Corollary~\ref{cor:lowBoundFromU}.
}
\caption{$O(1)$-approximation for third case}
\label{alg:case3}
\end{algorithm}

\begin{theorem}\label{thm:case3}
Algorithm~\ref{alg:case3} is a
$7$-approximation
for the third case.
\end{theorem}
\ifbool{shortVersion}{}{%
\begin{proof}
If $\val(U_1)\geq \frac{1}{7} \nu^*$,
then Algorithm~\ref{alg:case3} is clearly
a $7$-approximation
since $\nu^*$ upper bounds OPT.
Hence, assume
\begin{equation}\label{eq:inElse}
\val(U_1)< \frac{1}{7}\nu^*.
\end{equation}
Notice that the slope from $p_1$ to $p_2$
is not larger than
the one from the origin to $p_2$, i.e.,
\begin{equation}\label{eq:compSlopes}
\frac{\val(U_2)-\val(U_1)}{c(U_2)-c(U_1)}
  \leq \frac{\val(U_2)}{c(U_2)}.
\end{equation}
We therefore obtain
\begin{align*}
\val(R) &\geq \frac{1}{6}
  \cdot B \cdot \frac{\val(U_2)}{c(U_2)} \\
   &\geq
      \frac{1}{6} \cdot B \cdot
      \frac{\val(U_2)-\val(U_1)}{c(U_2)-c(U_1)}
      && \text{(using \eqref{eq:compSlopes})}    \\
   &\geq \frac{1}{6}\cdot 
        (B-c(U_1))
      \frac{\val(U_2)-\val(U_1)}{c(U_2)-c(U_1)}\\
   &=
      \frac{1}{6}\cdot (\nu^*-\val(U_1))
      && \text{(using \eqref{eq:defFStar})} \\
   &>
      \frac{1}{7} \nu^*,
      && \text{(using \eqref{eq:inElse})}.
\end{align*}
\end{proof}
}

Thus, it remains to show
Theorem~\ref{thm:mainAlgThm}.
Finally, our main result, Theorem~\ref{thm:MSTmain}, is a direct
consequence of the fact that we have a $7$-approximation
for all three cases under the assumption that each weight is either 
zero or a power of two. Hence, this implies a $14$-approximation
for general weights.

\section{Algorithm proving Theorem~\ref{thm:mainAlgThm}\label{sec:mainAlg}}

In this section, we present an algorithm that proves
Theorem~\ref{thm:mainAlgThm}.
For brevity, we define $[k]=\{1,\dots, k\}$ for
$k\in \mathbb{Z}_{\geq 0}$; in particular,
$[0]=\emptyset$.
Throughout this section let
$U\subseteq E_{\leq p-1}$ with $c(U)>B$.
Furthermore, for
$i\in \{-1,\dots, p\}$, we define
\begin{equation*}
U_{\leq i} = U \cap E_{\leq i}.
\end{equation*}

For each $i\in \{-1,\dots, p\}$, let
$\mathcal{A}_i\subseteq 2^V$ be the partition
of $V$ that corresponds to the connected components
of $(V,E_{\leq i}\setminus U)$. 
Notice that the partitions $\mathcal{A}_i$ become
coarser with increasing index $i$. Furthermore,
$\mathcal{A}_p=\{V\}$, since we assume that $(V,E_{p})$
is connected and $U$ does not contain any edges of
$E_p$. See Figure~\ref{fig:partition1} for an example.
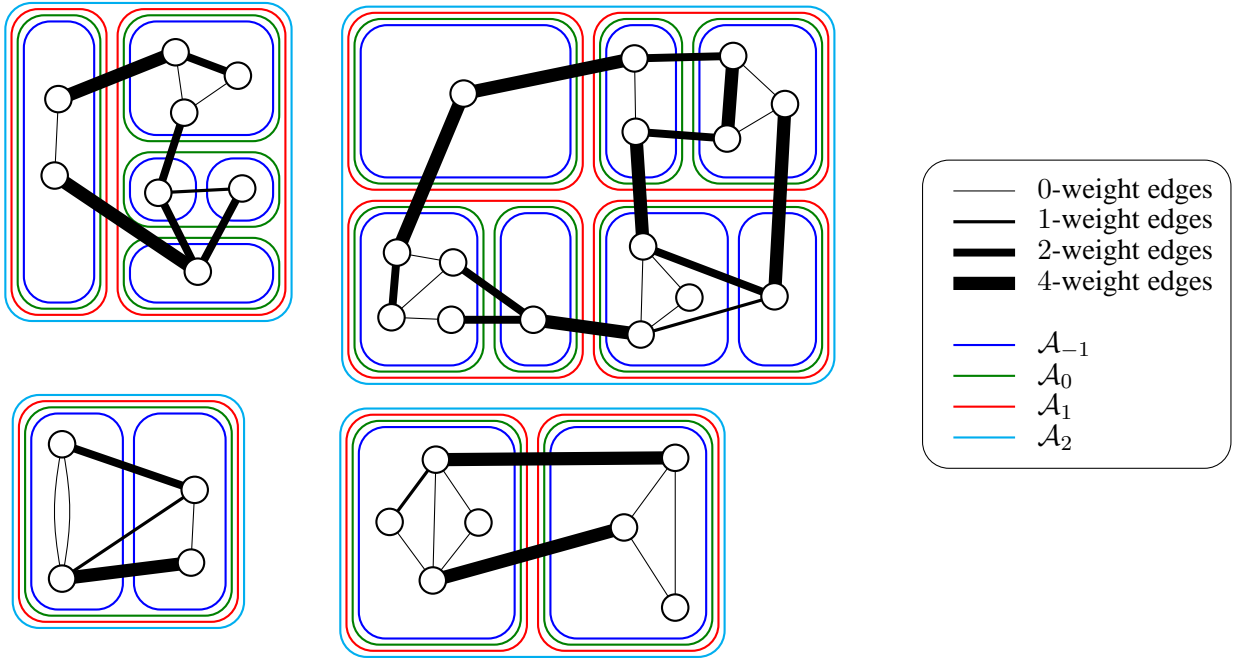
\begin{figure}
\begin{center}
\pgfdeclarelayer{background}
\pgfdeclarelayer{foreground}
\pgfsetlayers{background,main,foreground}

\begin{tikzpicture}[scale=0.82]

\begin{pgfonlayer}{foreground}
\begin{scope}[every node/.style={thick,draw=black,fill=white,circle,minimum size=10, inner sep=2pt}]
\node  (1) at (1.49,-1.51) {};
\node  (2) at (1.43,-2.75) {};
\node  (3) at (3.39,-0.75) {};
\node  (4) at (4.40,-1.13) {};
\node  (5) at (3.53,-1.73) {};
\node  (6) at (3.11,-3.03) {};
\node  (7) at (4.47,-2.96) {};
\node  (8) at (3.75,-4.31) {};
\node  (9) at (8.06,-1.42) {};
\node (10) at (10.83,-0.85) {};
\node (11) at (10.85,-2.04) {};
\node (12) at (12.43,-0.81) {};
\node (13) at (13.27,-1.59) {};
\node (14) at (12.33,-2.15) {};
\node (15) at (6.98,-4.00) {};
\node (16) at (7.89,-4.16) {};
\node (17) at (6.89,-5.05) {};
\node (18) at (7.86,-5.09) {};
\node (19) at (9.19,-5.09) {};
\node (20) at (10.97,-3.90) {};
\node (21) at (11.72,-4.73) {};
\node (22) at (10.93,-5.33) {};
\node (23) at (13.10,-4.71) {};
\node (24) at (1.55,-7.11) {};
\node (25) at (1.55,-9.29) {};
\node (26) at (3.70,-7.85) {};
\node (27) at (3.64,-9.03) {};
\node (28) at (7.61,-7.36) {};
\node (29) at (6.86,-8.37) {};
\node (30) at (8.30,-8.38) {};
\node (31) at (7.56,-9.32) {};
\node (32) at (11.49,-7.33) {};
\node (33) at (10.66,-8.46) {};
\node (34) at (11.49,-9.76) {};
\end{scope}

\end{pgfonlayer}

\tikzstyle{w0} = [black]
\tikzstyle{w1} = [black, line width=1.2pt]
\tikzstyle{w2} = [black, line width=3pt]
\tikzstyle{w4} = [black, line width=5pt]

\tikzstyle{Am1}=[thick,blue]
\tikzstyle{A0}=[thick,green!50!black]
\tikzstyle{A1}=[thick,red]
\tikzstyle{A2}=[thick,cyan]

\begin{scope}[shift={(16,-3)}] 
\def\sep{0.5}
\draw[w0] (0,0) -- (1,0);
\node[right] at (1.2,0) {$0$-weight edges};

\draw[w1] (0,-\sep) -- (1,-\sep);
\node[right] at (1.2,-\sep) {$1$-weight edges};

\draw[w2] (0,-2*\sep) -- (1,-2*\sep);
\node[right] at (1.2,-2*\sep) {$2$-weight edges};

\draw[w4] (0,-3*\sep) -- (1,-3*\sep);
\node[right] at (1.2,-3*\sep) {$4$-weight edges};

\draw[Am1] (0,-5*\sep) -- (1,-5*\sep);
\node[right] at (1.2,-5*\sep) {$\mathcal{A}_{-1}$};

\draw[A0] (0,-6*\sep) -- (1,-6*\sep);
\node[right] at (1.2,-6*\sep) {$\mathcal{A}_{0}$};

\draw[A1] (0,-7*\sep) -- (1,-7*\sep);
\node[right] at (1.2,-7*\sep) {$\mathcal{A}_{1}$};

\draw[A2] (0,-8*\sep) -- (1,-8*\sep);
\node[right] at (1.2,-8*\sep) {$\mathcal{A}_{2}$};

\draw[rounded corners=10pt] (-\sep,\sep) rectangle (4.5,-4.5);

\end{scope}

\begin{scope}[shorten <= -2pt, shorten >= -2pt]

\begin{scope}[w0] 
\draw  (1) --  (2);

\draw[w2]  (3) --  (4);
\draw  (3) --  (5);
\draw  (4) --  (5);

\draw (10) -- (11);

\draw (12) -- (13);
\draw[w4] (12) -- (14);
\draw (13) -- (14);

\draw (15) -- (16);
\draw[w2] (15) -- (17);
\draw (16) -- (17);
\draw (17) -- (18);

\draw (20) -- (21);
\draw (20) -- (22);
\draw (21) -- (22);

\draw (24) to[bend left=10] (25);
\draw (24) to[bend right=10] (25);

\draw (26) -- (27);

\draw[w1] (28) -- (29);
\draw (28) -- (30);
\draw (28) -- (31);
\draw (29) -- (31);
\draw (30) -- (31);

\draw (32) -- (33);
\draw (32) -- (34);
\draw (33) -- (34);
\end{scope}

\begin{scope}[w1] 
\draw  (6) --  (7);

\draw[w2] (20) -- (23);
\draw (22) -- (23);

\draw[w2] (24) -- (26);
\draw (25) -- (26);
\draw[w4] (25) -- (27);
\end{scope}

\begin{scope}[w2] 
\draw  (5) --  (6);
\draw  (6) --  (8);
\draw  (7) --  (8);

\draw (10) -- (12);
\draw (11) -- (14);

\draw (16) -- (19);
\draw (18) -- (19);
\end{scope}

\begin{scope}[w4] 
\draw  (1) --  (3);
\draw  (2) --  (8);
\draw  (9) -- (10);
\draw  (11) -- (20);
\draw (13) -- (23);
\draw (15) --  (9);
\draw (19) -- (22);
\draw (28) -- (32);
\draw (31) -- (33);
\end{scope}

\end{scope}

\begin{pgfonlayer}{background}

\begin{scope} 

\def\sep{0.1} 
\def\msep{0.09} 
\def\exp{0.8}

\begin{scope}[A2]

\path let 
    \p1=(2),
    \p2=(3) in coordinate (r1tl) at ($(\x1,\y2)+(-\exp,\exp)$);
\path let 
    \p1=(8),
    \p2=(7) in coordinate (r1br) at ($(\x2,\y1)+(\exp,-\exp)$);

\draw[rounded corners=10pt] (r1tl) rectangle (r1br);

\path let 
    \p1=(17),
    \p2=(12) in coordinate (r2tl) at ($(\x1,\y2)+(-\exp,\exp)$);
\path let 
    \p1=(22),
    \p2=(13) in coordinate (r2br) at ($(\x2,\y1)+(\exp,-\exp)$);

\draw[rounded corners=10pt] (r2tl) rectangle (r2br);

\path let 
    \p1=(25),
    \p2=(26) in coordinate (r3br) at ($(\x2,\y1)+(\exp,-\exp)$);
\path coordinate (r3tl) at ($(24)+(-\exp,\exp)$);

\draw[rounded corners=10pt] (r3tl) rectangle (r3br);

\path let 
    \p1=(29),
    \p2=(32) in coordinate (r4tl) at ($(\x1,\y2)+(-\exp,\exp)$);
\path let 
    \p1=(34),
    \p2=(34) in coordinate (r4br) at ($(\x2,\y1)+(\exp,-\exp)$);

\draw[rounded corners=10pt] (r4tl) rectangle (r4br);

\end{scope}

\begin{scope}[A1]

\coordinate  (r11tl) at ($(r1tl)+(\sep,-\sep)$);
\path let
  \p1=(r1tl),
  \p2=(r1br) in coordinate (r11br) at ($(0.4*\x1+0.4*\x2,\y2) + (-\msep,\sep)$);

\draw[rounded corners=10pt] (r11tl) rectangle (r11br);

\path let
  \p1=(r1tl),
  \p2=(r1br) in coordinate (r12tl) at ($(0.4*\x1+0.4*\x2,\y1) + ( \msep,-\sep)$);
\coordinate  (r12br) at ($(r1br)+(-\sep,\sep)$);

\draw[rounded corners=10pt] (r12tl) rectangle (r12br);

\coordinate (r2xmid) at ($(r2tl)!0.5!(r2br) + (0,0)$);

\coordinate  (r21tl) at ($(r2tl)+(\sep,-\sep)$);
\coordinate (r21br) at ($(r2xmid)+(-\msep,\msep)$);

\draw[rounded corners=10pt] (r21tl) rectangle (r21br);

\path let
  \p1 = (r2tl),
  \p2 = (r2xmid) in
     coordinate (r22tl) at ($(\x2,\y1)+(\msep,-\sep)$);

\path let
  \p1 = (r2xmid),
  \p2 = (r2br) in
     coordinate (r22br) at ($(\x2,\y1)+(-\sep,\msep)$);

\draw[rounded corners=10pt] (r22tl) rectangle (r22br);

\path let
  \p1 = (r2tl),
  \p2 = (r2xmid) in
     coordinate (r23tl) at ($(\x1,\y2)+(\sep,-\msep)$);

\path let
  \p1 = (r2xmid),
  \p2 = (r2br) in
     coordinate (r23br) at ($(\x1,\y2)+(-\msep,\sep)$);

\draw[rounded corners=10pt] (r23tl) rectangle (r23br);

\coordinate  (r24tl) at ($(r2xmid) + (\msep,-\msep)$);
\coordinate (r24br) at ($(r2br)+(-\sep,\sep)$);

\draw[rounded corners=10pt] (r24tl) rectangle (r24br);

\coordinate  (r31tl) at ($(r3tl)+(\sep,-\sep)$);
\coordinate  (r31br) at ($(r3br)+(-\sep,\sep)$);

\draw[rounded corners=10pt] (r31tl) rectangle (r31br);

\coordinate (r4xmid) at ($(r4tl)!0.5!(r4br)$);

\coordinate (r41tl) at ($(r4tl)+(\sep,-\sep)$);

\path let
  \p1 = (r4xmid),
  \p2 = (r4br) in
     coordinate (r41br) at ($(\x1,\y2)+(-\msep,\sep)$);

\draw[rounded corners=10pt] (r41tl) rectangle (r41br);

\path let
  \p1 = (r4xmid),
  \p2 = (r4tl) in
     coordinate (r42tl) at ($(\x1,\y2)+(\msep,-\sep)$);

\coordinate (r42br) at ($(r4br)+(-\sep,\sep)$);

\draw[rounded corners=10pt] (r42tl) rectangle (r42br);

\end{scope}

\begin{scope}[A0]

\coordinate  (r111tl) at ($(r11tl)+(\sep,-\sep)$);
\coordinate  (r111br) at ($(r11br)+(-\sep,\sep)$);

\draw[rounded corners=10pt] (r111tl) rectangle (r111br);

\coordinate (r12xa) at ($(r12tl)!0.45!(r12br)$);
\coordinate (r12xb) at ($(r12tl)!0.73!(r12br)$);

\coordinate (r121tl) at ($(r12tl)+(\sep,-\sep)$);
\path let
  \p1 = (r12xa),
  \p2 = (r12br) in
     coordinate (r121br) at ($(\x2,\y1)+(-\sep,\msep)$);

\draw[rounded corners=10pt] (r121tl) rectangle (r121br);

\path let
  \p1 = (r12xa),
  \p2 = (r12tl) in
     coordinate (r122tl) at ($(\x2,\y1)+(\sep,-\msep)$);

\path let
  \p1 = (r12xb),
  \p2 = (r12br) in
     coordinate (r122br) at ($(\x2,\y1)+(-\sep,\msep)$);

\draw[rounded corners=10pt] (r122tl) rectangle (r122br);

\path let
  \p1 = (r12xb),
  \p2 = (r12tl) in
     coordinate (r123tl) at ($(\x2,\y1)+(\sep,-\msep)$);

\coordinate (r123br) at ($(r12br)+(-\sep,\sep)$);

\draw[rounded corners=10pt] (r123tl) rectangle (r123br);

\coordinate  (r211tl) at ($(r21tl)+(\sep,-\sep)$);
\coordinate  (r211br) at ($(r21br)+(-\sep,\sep)$);

\draw[rounded corners=10pt] (r211tl) rectangle (r211br);

\coordinate (r22x) at ($(r22tl)!0.4!(r22br)$);

\coordinate (r221tl) at ($(r22tl)+(\sep,-\sep)$);
\path let
  \p1 = (r22x),
  \p2 = (r22br) in
     coordinate (r221br) at ($(\x1,\y2)+(-\msep,\sep)$);

\draw[rounded corners=10pt] (r221tl) rectangle (r221br);

\path let
  \p1 = (r22x),
  \p2 = (r22tl) in
     coordinate (r222tl) at ($(\x1,\y2)+(\msep,-\sep)$);
\coordinate (r222br) at ($(r22br)+(-\sep,\sep)$);

\draw[rounded corners=10pt] (r222tl) rectangle (r222br);

\coordinate (r23x) at ($(r23tl)!0.6!(r23br)$);

\coordinate (r231tl) at ($(r23tl)+(\sep,-\sep)$);
\path let
  \p1 = (r23x),
  \p2 = (r23br) in
     coordinate (r231br) at ($(\x1,\y2)+(-\msep,\sep)$);

\draw[rounded corners=10pt] (r231tl) rectangle (r231br);

\path let
  \p1 = (r23x),
  \p2 = (r23tl) in
     coordinate (r232tl) at ($(\x1,\y2)+(\msep,-\sep)$);
\coordinate (r232br) at ($(r23br)+(-\sep,\sep)$);

\draw[rounded corners=10pt] (r232tl) rectangle (r232br);

\coordinate  (r241tl) at ($(r24tl)+(\sep,-\sep)$);
\coordinate  (r241br) at ($(r24br)+(-\sep,\sep)$);

\draw[rounded corners=10pt] (r241tl) rectangle (r241br);

\coordinate  (r311tl) at ($(r31tl)+(\sep,-\sep)$);
\coordinate  (r311br) at ($(r31br)+(-\sep,\sep)$);

\draw[rounded corners=10pt] (r311tl) rectangle (r311br);

\coordinate  (r411tl) at ($(r41tl)+(\sep,-\sep)$);
\coordinate  (r411br) at ($(r41br)+(-\sep,\sep)$);

\draw[rounded corners=10pt] (r411tl) rectangle (r411br);

\coordinate  (r421tl) at ($(r42tl)+(\sep,-\sep)$);
\coordinate  (r421br) at ($(r42br)+(-\sep,\sep)$);

\draw[rounded corners=10pt] (r421tl) rectangle (r421br);

\end{scope}

\begin{scope}[Am1]

\coordinate  (r1111tl) at ($(r111tl)+(\sep,-\sep)$);
\coordinate  (r1111br) at ($(r111br)+(-\sep,\sep)$);

\draw[rounded corners=10pt] (r1111tl) rectangle (r1111br);

\coordinate  (r1211tl) at ($(r121tl)+(\sep,-\sep)$);
\coordinate  (r1211br) at ($(r121br)+(-\sep,\sep)$);

\draw[rounded corners=10pt] (r1211tl) rectangle (r1211br);

\coordinate (r122x) at ($(r122tl)!0.5!(r122br)$);

\coordinate (r1221tl) at ($(r122tl)+(\sep,-\sep)$);
\path let 
  \p1 = (r122x),
  \p2 = (r122br) in
     coordinate (r1221br) at ($(\x1,\y2)+(-\msep,\sep)$);

\draw[rounded corners=10pt] (r1221tl) rectangle (r1221br);

\path let 
  \p1 = (r122x),
  \p2 = (r122tl) in
     coordinate (r1222tl) at ($(\x1,\y2)+(\msep,-\sep)$);
\coordinate (r1222br) at ($(r122br)+(-\sep,\sep)$);

\draw[rounded corners=10pt] (r1222tl) rectangle (r1222br);

\coordinate  (r1231tl) at ($(r123tl)+(\sep,-\sep)$);
\coordinate  (r1231br) at ($(r123br)+(-\sep,\sep)$);

\draw[rounded corners=10pt] (r1231tl) rectangle (r1231br);

\coordinate  (r2111tl) at ($(r211tl)+(\sep,-\sep)$);
\coordinate  (r2111br) at ($(r211br)+(-\sep,\sep)$);

\draw[rounded corners=10pt] (r2111tl) rectangle (r2111br);

\coordinate  (r2211tl) at ($(r221tl)+(\sep,-\sep)$);
\coordinate  (r2211br) at ($(r221br)+(-\sep,\sep)$);

\draw[rounded corners=10pt] (r2211tl) rectangle (r2211br);

\coordinate  (r2221tl) at ($(r222tl)+(\sep,-\sep)$);
\coordinate  (r2221br) at ($(r222br)+(-\sep,\sep)$);

\draw[rounded corners=10pt] (r2221tl) rectangle (r2221br);

\coordinate  (r2311tl) at ($(r231tl)+(\sep,-\sep)$);
\coordinate  (r2311br) at ($(r231br)+(-\sep,\sep)$);

\draw[rounded corners=10pt] (r2311tl) rectangle (r2311br);

\coordinate  (r2321tl) at ($(r232tl)+(\sep,-\sep)$);
\coordinate  (r2321br) at ($(r232br)+(-\sep,\sep)$);

\draw[rounded corners=10pt] (r2321tl) rectangle (r2321br);

\coordinate (r241x) at ($(r241tl)!0.6!(r241br)$);

\coordinate (r2411tl) at ($(r241tl)+(\sep,-\sep)$);
\path let 
  \p1 = (r241x),
  \p2 = (r241br) in
     coordinate (r2411br) at ($(\x1,\y2)+(-\msep,\sep)$);

\draw[rounded corners=10pt] (r2411tl) rectangle (r2411br);

\path let 
  \p1 = (r241x),
  \p2 = (r241tl) in
     coordinate (r2412tl) at ($(\x1,\y2)+(\msep,-\sep)$);
\coordinate (r2412br) at ($(r241br)+(-\sep,\sep)$);

\draw[rounded corners=10pt] (r2412tl) rectangle (r2412br);

\coordinate  (r3111tl) at ($(r311tl)+(\sep,-\sep)$);
\path let
  \p1=(r311tl),
  \p2=(r311br) in coordinate (r3111br) at ($(0.5*\x1+0.5*\x2,\y2) + (- \msep,\sep)$);

\draw[rounded corners=10pt] (r3111tl) rectangle (r3111br);

\path let
  \p1=(r311tl),
  \p2=(r311br) in coordinate (r3112tl) at ($(0.5*\x1+0.5*\x2,\y1) + ( \msep,-\sep)$);
\coordinate  (r3112br) at ($(r311br)+(-\sep,\sep)$);

\draw[rounded corners=10pt] (r3112tl) rectangle (r3112br);

\coordinate  (r4111tl) at ($(r411tl)+(\sep,-\sep)$);
\coordinate  (r4111br) at ($(r411br)+(-\sep,\sep)$);

\draw[rounded corners=10pt] (r4111tl) rectangle (r4111br);

\coordinate  (r4211tl) at ($(r421tl)+(\sep,-\sep)$);
\coordinate  (r4211br) at ($(r421br)+(-\sep,\sep)$);

\draw[rounded corners=10pt] (r4211tl) rectangle (r4211br);

\end{scope}

\end{scope} 
\end{pgfonlayer}

\end{tikzpicture}

\end{center}
\caption{Example of a graph $(V,E_{\leq p-1}\setminus U)$
for $p=3$ 
together with its corresponding partitions
$\mathcal{A}_{-1}$, $\mathcal{A}_0$,
$\mathcal{A}_1$, and $\mathcal{A}_2$. The
edges in $E_p$, which connect all vertices
by assumption, and the coarsest partition
$\mathcal{A}_3=\{V\}$ are not shown.
}
\label{fig:partition1}
\end{figure}
For
$i\in \{-1,\dots, p\}$ and $A\in \mathcal{A}_i$,
we denote by $\mathcal{C}_i(A)\subseteq \mathcal{A}_{i-1}$
the sets in $\mathcal{A}_{i-1}$ that are included in $A$,
which we call the \emph{children} of $A$ (on level $i$).
More formally:
\begin{equation*}
\mathcal{C}_i(A) = \begin{cases}
\emptyset &\text{if } i=-1,\\
\{C\in \mathcal{A}_{i-1} \mid C\subseteq A\}
   &\text{if } i\geq 0.
\end{cases}
\end{equation*}
Notice that when talking about children,
we must indicate on which level $i$ we
consider the set $A$, since $A$ may be a set that
exists in several consecutive partitions.
In this case, one has $\mathcal{C}_i(A) = \{A\}$
for all levels $i$ such that $A\in \mathcal{A}_i$,
except for the most fine-grained one (smallest
$i$ such that $A\in \mathcal{A}_i$).

Our algorithm greedily constructs what we call
a removal pattern.
\begin{definition}[Removal pattern]
A \emph{removal pattern}
$\mathcal{W} = \{(W_1,i_1), \dots, (W_\beta, i_\beta)\}$
is a family of tuples, where
$i_1,\dots, i_\beta \in \{-1,\dots, p-1\}$,
$W_q\in \mathcal{A}_{i_q}$ for $q\in [\beta]$, and
$W_1,\dots, W_\beta$ are all disjoint sets.
\end{definition}
To each removal pattern we assign a set of corresponding
edges $R(\mathcal{W})$ to be removed:
\begin{equation*}
R(\mathcal{W}) = \bigcup_{(W,i)\in \mathcal{W}}
 \{e\in U_{\leq i} \mid |e\cap W|=1\},
\end{equation*}
where $|e\cap W|$ counts the number of endpoint
that $e$ has in $W$. In general, we treat an
edge $e=\{u,v\}$
as a set containing its two endpoints $u$
and $v$.

The motivation for the use of a removal pattern $\mathcal{W}$
to define an interdiction set, is
that when removing all edges $U_{\leq i}$ that touch $W_q$,
we have locally the same impact on the levels
$-1,\dots, i$ as $U$ has when removing it from the
graph. This allows us to exploit synergies between different
levels that exist when removing $U$.
For notational convenience, we denote the cost
of the edges $R(\mathcal{W})$ that
correspond to $\mathcal{W}$ by
\begin{equation*}
c(\mathcal{W}) = c(R(\mathcal{W})).
\end{equation*}

To decide which sets to add to $\mathcal{W}$, we define for
$i\in \{-1,\dots, p\}$ an auxiliary cost function
$\kappa_i: \mathcal{A}_i \rightarrow \mathbb{Z}_{\geq 0}$
and an auxiliary impact function
$g_i: \mathcal{A}_i \rightarrow \mathbb{Z}_{\geq 0}$
as follows: Let $A\in \mathcal{A}_i$, then
\begin{align*}
\kappa_i(A) &=   c(\{e\in U_{\leq i} \mid |e\cap A|=1\})
                +2c(\{e\in U_{\leq i} \mid |e\cap A|=2\}),\\
g_i(A)      &= |\{D\in \mathcal{A}_{-1} \mid D \subseteq A\}| +
          \sum_{\ell=0}^{i} 2^{\ell} \cdot
         |\{D\in \mathcal{A}_{\ell}\mid D\subseteq A\}|.
\end{align*}
Notice that for $i\in \{0,\dots, p\}$,
\begin{align}
\kappa_i(A) &\geq 
 \sum_{C\in \mathcal{C}_i(A)} \kappa_{i-1}(C), \text{ and}
\label{eq:kappaRecursive}\\
g_i(A)  &= 2^{i} + |\{D\in \mathcal{A}_{-1} \mid D\subseteq A\}|
         + \sum_{\ell=0}^{i-1} 2^{\ell}
         |\{D\in \mathcal{A}_{\ell}\mid D\subseteq A\}|
        = 2^{i} + \sum_{C\in\mathcal{C}_i(A)}
            g_{i-1}(C).\label{eq:gRecursive}
\end{align}
These recursive relations are a main reason
why we use $g_i$ and $\kappa_i$ as proxys for measuring
locally the impact and cost of the removal set $U$.
Moreover we have the following basic properties.
\begin{lemma}
\begin{align}
g_p(V)-2^{p+1} &= \val(U),\label{eq:gpV}\\
\kappa_p(V) &= 2c(U).\label{eq:kappapV}
\end{align}
\end{lemma}
\ifbool{shortVersion}{}{%
\begin{proof}
Equation~\eqref{eq:gpV} holds since 
\begin{align*}
g_p(V) &= |\{D\in \mathcal{A}_{-1} \mid D\subseteq V\}|
  + \sum_{\ell=0}^{p} 2^\ell \cdot
    |\{D\in \mathcal{A}_\ell \mid D\subseteq V\}|\\
  &= |\mathcal{A}_{-1}|
    + \sum_{\ell=0}^p 2^\ell\cdot |\mathcal{A}_\ell|\\
  &= \sigma(E_{-1}\setminus U)
    + \sum_{\ell=0}^p 2^\ell\cdot \sigma(E_{\leq \ell}\setminus U)\\
  &= (\sigma(E_{-1}\setminus U)-1)
    + \sum_{\ell=0}^p 2^\ell\cdot
       (\sigma(E_{\leq \ell}\setminus U)-1) + 2^{p+1}\\
  &= \val(U) + 2^{p+1}.
\end{align*}
Furthermore,~\eqref{eq:kappapV} follows immediately
from the definition of $\kappa_i$ and the
observation that $U_{\leq p} = U$:
\begin{align*}
\kappa_p(V) =  c(%
\underbrace{\{e\in U_{\leq p} \mid e\cap V = 1\}}_{=\emptyset})
 + 2c(\underbrace{\{e\in U_{\leq p} \mid |e\cap V|=2\}}_{=U_{\leq p}})
  = 2c(U).
\end{align*}

\end{proof}
}

For $i\in \{-1,\dots, p\}$ and $A\in \mathcal{A}_i$,
we define the \emph{auxiliary efficiency} of $A$ by
\begin{equation*}
\rho_i(A) = \frac{g_i(A)}{\kappa_i(A)},
\end{equation*}
with the convention that $\rho_i(A)=\infty$ if
$\kappa_i(A)=0$.
Our algorithm, as described in
Algorithm~\ref{alg:mainAlg}, adds sets to $\mathcal{W}$
iteratively starting at level $p-1$ and descending to
level $-1$. Among the sets considered in each level,
preference is given to sets with higher auxiliary
efficiency.
In the following we will show that the interdiction
set $R(\mathcal{W})$ returned by Algorithm~\ref{alg:mainAlg}
satisfies the conditions of Theorem~\ref{thm:mainAlgThm}.

\begin{algorithm}[h]
\DontPrintSemicolon
$\mathcal{W}=\emptyset$\;
$\ell=p-1$  \tcp*{current level}
$A=V$ \tcp*{current vertex set to break into
components on levels $\leq \ell$}
\While{$\ell\neq -2$}{
Let $\mathcal{C}_{\ell+1}(A) = \{Q_1,\dots, Q_h\}$
, where the numbering is chosen such that
\vspace*{-0.5em}
\begin{equation*}
\rho_\ell(Q_1) \geq \rho_\ell(Q_2) \geq
  \dots \geq \rho_\ell(Q_h).
\end{equation*}\;
\vspace*{-1.5em}
Let
\vspace*{-0.5em}
\begin{equation*}
s =
  \max\left\{j\in \{0,\dots, h\}\;\Bigg\vert\;
     c\left(\mathcal{W}\cup
         \{(Q_k, \ell) \mid k\in [j]\}
           \right) \leq B\right\}.
\end{equation*}\;
\vspace*{-1.5em}
Set
\begin{equation*}
\mathcal{W} = \mathcal{W} \cup
  \{(Q_k,\ell) \mid k\in [s]\}.
\end{equation*}\;
\vspace*{-1.5em}
\eIf{$s < h$}{
$\ell = \ell - 1$\;
$A=Q_{s+1}$\;
}{
$\ell=-2$ \quad (i.e., leave the while-loop) 
}
}
\Return $R(\mathcal{W})$\;

\caption{Construction of interdiction set $R$
fulfilling conditions of Theorem~\ref{thm:mainAlgThm}.}
\label{alg:mainAlg}
\end{algorithm}

\section{Analysis of the algorithm}\label{sec:analysis}

We first formalize a particular structure of the removal pattern
returned by Algorithm~\ref{alg:mainAlg} which follows
immediately from the fact that Algorithm~\ref{alg:mainAlg}
considers elements to add to $\mathcal{W}$ with respect
to decreasing order of their auxiliary efficiencies.

\begin{definition}[efficient removal pattern]
\label{def:effRemovalPattern}
Let $\mathcal{W}$ be a removal pattern.
$\mathcal{W}$ is called \emph{efficient} if for
every $i\in \{0,\dots, p\}$ and $A\in \mathcal{A}_i$,
one of the following holds:
\vspace{-0.5em}
\begin{enumerate}[(i)]
\setlength\itemsep{-0.5em}
\item\label{item:ERPNoDesc}
No descendant of $A$ is contained in $\mathcal{W}$, i.e.,
for every $\ell \in \{-1,\dots, i-1\}$ and $D\in \mathcal{A}_{\ell}$
with $D\subseteq A$, we have $(D,\ell)\not\in \mathcal{W}$, or

\item\label{item:ERPDesc}
all sets $(W,i')\in \mathcal{W}$ for $i'\in \{-1,\dots, i-1\}$
are descendants of $(A,i)$.
Moreover, there is
a numbering of the elements in 
$\mathcal{C}_i(A)$, say $\mathcal{C}_i(A)=\{Q_1,\dots, Q_h\}$,
and $s \in \{0,\dots, h\}$ such that
$\rho_i(Q_1)\geq \dots \geq \rho_i(Q_h)$ and the following
holds:

\vspace{-0.5em}
\begin{itemize}
\item $(Q_k,i-1)\in \mathcal{W}$ for $k\in \{1,\dots, s\}$,

\item $(Q_k,i-1)\not\in \mathcal{W}$ for $k\in \{s+1,\dots, h\}$,

\item
all tuples in $\mathcal{W}$ on levels $\{-1,\dots, i-2\}$
are descendants of $(Q_{s+1},i-1)$.
In particular, if $s=h$, then $\mathcal{W}$ contains no
tuples on levels $\{-1,\dots, i-2\}$.
\end{itemize}
\end{enumerate}
\end{definition}

Clearly, Algorithm~\ref{alg:mainAlg} returns an
efficient removal pattern.
The key motivation for concentrating on efficient removal patterns
is that we can relate, for any efficient removal pattern $\mathcal{W}$,
its corresponding value $\val(R)$, where $R=R(\mathcal{W})$,
to its cost $c(R)$.
To do so, we first introduce variants $\kappa_i^\mathcal{W}$ and
$g_i^{\mathcal{W}}$ of the auxiliary cost and impact
functions $\kappa_i$ and $g_i$, that measure cost and impact
of the efficient removal pattern $\mathcal{W}$.
In what follows, let $\mathcal{W}$ be an efficient removal pattern
with corresponding removal set $R=R(\mathcal{W})$.

As usual we use the notation
$R_{\leq i}=R\cap E_{\leq i}$ for
$i\in \{-1,\dots, p\}$.
For $\ell\in \{-1,\dots, p\}$ we define
$\mathcal{S}_\ell\subseteq \mathcal{A}_{\ell}$
to be all sets of $\mathcal{A}_\ell$ that are
descendants of sets added to $\mathcal{W}$, i.e., 
\begin{equation*}
\mathcal{S}_\ell = \{A\in \mathcal{A}_\ell \mid
\exists (W,i)\in \mathcal{W} \text{ with }
i\geq \ell \text{ and } A\subseteq W\}.
\end{equation*}
Notice that contrary to $\mathcal{A}_\ell$, the family
$\mathcal{S}_\ell$ is generally not a partition.

Similarly to the definitions of the auxiliary impact
function $g_i$ and auxiliary cost function
$\kappa_i$, which are defined in terms of the
set $U$, we define corresponding functions
$g^\mathcal{W}_i$ and
$\kappa^{\mathcal{W}}_i$ for the efficient removal
pattern $\mathcal{W}$.
For $i\in \{-1,\dots, p\}$
and $A\in \mathcal{A}_i$, let
\begin{align*}
\kappa_i^{\mathcal{W}}(A)&=
 \sum_{\substack{(W,j)\in \mathcal{W} \text{ with}\\
    W\subseteq A, j\leq i}}
  \kappa_j(W). \\
g_i^{\mathcal{W}}(A) &=
 \sum_{\substack{(W,j)\in \mathcal{W} \text{ with}\\
    W\subseteq A, j\leq i}}
  g_j(W)
 = |\{S\in \mathcal{S}_{-1} \mid S\subseteq A\}| +
  \sum_{\ell=0}^i 2^{l}
  |\{S\in \mathcal{S}_\ell \mid S\subseteq A\}|.
\end{align*}

The functions $\kappa_i^{\mathcal{W}}$ and
$g_i^{\mathcal{W}}$ are thus analogous to
$\kappa_i$ and $g_i$ with the difference that
they only consider sets of the partitions $\mathcal{A}_i$
that are subsets of a set in the removal pattern
$\mathcal{W}$.
Since each edge in $R$ crosses at least one
of the sets in the efficient removal pattern
$\mathcal{W}$, we obtain
\begin{equation}\label{eq:kappaWpV}
\kappa^{\mathcal{W}}_p(V) \geq c(R).
\end{equation}

Notice that if $(A,i)\in \mathcal{W}$ then
$\kappa_i^{\mathcal{W}}(A) = \kappa_i(A)$
and
$g_i^{\mathcal{W}}(A) = g_i(A)$.
Furthermore, for $i\in \{0,\dots, p-1\}$
and $(A,i)\not\in \mathcal{W}$ we have
\begin{align}
\kappa_i^{\mathcal{W}}(A) &=
 \sum_{C\in \mathcal{C}_i(A)}
   \kappa_{i-1}^\mathcal{W}(C) \text{ and}%
  \label{eq:kappaRec}\\
g_i^{\mathcal{W}}(A) &=
 \sum_{C\in \mathcal{C}_i(A)}
   g^{\mathcal{W}}_{i-1}(C).
  \label{eq:gRec}
\end{align}

The following shows a basic lower bound on $\val(R)$
in terms of $g_i^\mathcal{W}$.
\begin{proposition}\label{prop:impGW}
Let $\mathcal{W}$ be an efficient removal pattern
and $R=R(\mathcal{W})$ the corresponding removal set.
Then
\begin{equation*}
\val(R) \geq g^{\mathcal{W}}_p(V) - 2^{p-1}%
.
\end{equation*}
\end{proposition}
\ifbool{shortVersion}{}{%
\begin{proof}

For each $i\in \{-1,\dots, p-1\}$, 
the number $\sigma(E_{\leq i}\setminus R)$
of connected components
of $(V, E_{\leq i}\setminus R)$
is at least $|\mathcal{S}_i|$, since
each $S\in \mathcal{S}_i$ is a connected
component of $(V,E_{\leq i}\setminus R)$.
Furthermore, only if $\mathcal{S}_i$ is
a partition of $V$ we have
$\sigma(E_{\leq i}\setminus U)=|\mathcal{S}_i|$,
otherwise there is at least one more connected
component in $(V,E_{\leq i}\setminus R)$,
and thus $\sigma(E_{\leq i}\setminus U)
> |\mathcal{S}_i|$.
Notice that $\mathcal{S}_{p-1}$ does not
form a partition of $V$, since this would
imply $R=U$ which contradicts
$c(R) \leq B < c(U)$. Hence,
$\sigma(E_{\leq p-1}\setminus U) > |\mathcal{S}_{p-1}|$
and we obtain
\begin{align*}
\val(R) &= \sigma(E_{-1}\setminus U) - 1
   + \sum_{i=0}^{p-1}
  2^{i} \cdot (\sigma(E_{\leq i}\setminus U) - 1) \\
     &\geq 2^{p-1}\cdot |\mathcal{S}_{p-1}|
   + |\mathcal{S}_{-1}| - 1
   + \sum_{i=0}^{p-2} 2^{i}\cdot (|\mathcal{S}_{i}|-1)\\
     &= |\mathcal{S}_{-1}| + \left(\sum_{i=0}^{p-1} 2^{i} \cdot
        |\mathcal{S}_i|\right) - 2^{p-1}\\
     &= g^{\mathcal{W}}_p(V) - 2^{p-1}.
\end{align*}
\end{proof}
}

The following lemma relates cost and
impact function for the sets $U$ and $R$.

\begin{lemma}\label{lem:gToGW}
Let $\mathcal{W}$ be an efficient removal set,
let $i\in \{-1,\dots, p\}$, and let
$A\in \mathcal{A}_i$ such that
$\kappa_i(A)>0$. Then
\begin{equation*}
\frac{\kappa_i^{\mathcal{W}}(A)}{\kappa_i(A)}
 \cdot \left( g_i(A) - 2^i \right)
\leq g_i^{\mathcal{W}}(A)
   + 2^{i}.
\end{equation*}
\end{lemma}

\ifbool{shortVersion}{}{%
To prove Lemma~\ref{lem:gToGW}, we need
the following basic result, which is proven in
Appendix~\ref{app:prefixDom}.
\begin{lemma}\label{lem:prefixDom}
Let $k\in \mathbb{Z}_{>0}$, and let
$a_j,b_j \geq 0$ for $j\in [k]$ be reals
satisfying
$\frac{a_1}{b_1}\geq \dots \geq \frac{a_k}{b_k}$,
where we interpret $\frac{a}{b}=\infty$ if $b=0$,
independent of whether $a=0$.
Let $\lambda\in [0,1]$.
Then for any $q\in [k]$ with
$\left(\sum_{j=1}^{q-1} b_j\right) + \lambda b_q >0$ 
we have 
\begin{equation*}
\frac{\sum_{j=1}^k a_j}{\sum_{j=1}^k b_j}
\leq \frac{\left(\sum_{j=1}^{q-1} a_j\right) + \lambda a_q}%
{\left( \sum_{j=1}^{q-1} b_j\right) + \lambda b_q}.
\end{equation*}
\end{lemma}

\begin{proof}[Proof of Lemma~\ref{lem:gToGW}]
Let $i\in \{-1,\dots, p\}$ and
$A\in \mathcal{A}_i$ such that
$\kappa_i(A)>0$.
The result trivially holds if
$\kappa_i^{\mathcal{W}}(A)=0$; we thus
assume $\kappa_i^{\mathcal{W}}(A)>0$.
We prove the lemma by induction on $i$,
starting at $i=-1$.
First observe that 
if $(A,i)\in \mathcal{W}$,
then $g_i^{\mathcal{W}}(A)=g_i(A)$ and
$\kappa_i^{\mathcal{W}}(A)=\kappa_i(A)$,
and the result follows trivially.
This observation also covers the base case
$i=-1$ of the induction as
$\kappa_{-1}^{\mathcal{W}}(A)>0$ implies 
$(A,-1)\in \mathcal{W}$.

Thus, we assume from now on $i>-1$ and
$(A,i)\not\in \mathcal{W}$.
Since $\kappa_i^{\mathcal{W}}(A)>0$, the efficient removal
pattern $\mathcal{W}$ contains at least one descendant
of $(A,i)$. Hence, point~\eqref{item:ERPDesc} of the definition
of an efficient removal pattern, i.e.,
Definition~\ref{def:effRemovalPattern}, holds for $A\in \mathcal{A}_i$.
Let $\mathcal{C}_i(A)=\{Q_1,\dots, Q_h\}$,
where the numbering is chosen according to 
Definition~\ref{def:effRemovalPattern},
and let $s\in \{0,\dots, h\}$ be the
index as claimed by Definition~\ref{def:effRemovalPattern}.

Using~\eqref{eq:gRecursive}, we deduce
\begin{align}
\frac{\kappa^{\mathcal{W}}_i(A)}{\kappa_i(A)}
\left( g_i(A) - 2^i \right) &=
\frac{\kappa^{\mathcal{W}}_i(A)}{\kappa_i(A)}
\sum_{j=1}^h g_{i-1}(Q_j)
&& \text{(by~\eqref{eq:gRecursive})}
\notag\\
&\leq \frac{\kappa_i^{\mathcal{W}}(A)}%
 {\sum_{j=1}^h\kappa_{i-1}(Q_j)}
\sum_{j=1}^h g_{i-1}(Q_j).
  && \text{(by~\eqref{eq:kappaRecursive})}
\label{eq:gOverK1}
\end{align}

In a next step we will apply Lemma~\ref{lem:prefixDom}
with parameters $q=\min\{s+1,h\}$ and 
$\lambda={\kappa_{i-1}^{\mathcal{W}}(Q_{q})}/%
{\kappa_{i-1}(Q_{q})}$
to the ratio $\sum_{j=1}^h g_{i-1}(Q_j)/%
\sum_{j=1}^h \kappa_{i-1} (Q_j)$
in~\eqref{eq:gOverK1}, i.e., the terms
in the terminology of Lemma~\ref{lem:prefixDom} are
$a_j=g_{i-1}(Q_j)$
and $b_j=\kappa_{i-1}(Q_j)$ for $j\in [h]$.
To do so, we first check that the conditions
of Lemma~\ref{lem:prefixDom} are fulfilled. More precisely,
we have to show that:
\vspace{-0.5em}
\begin{enumerate}[(i)]
\setlength{\itemsep}{-0.3em}
\item $\lambda$ is well defined, i.e., $\kappa_{i-1}(Q_q)>0$,

\item $\lambda\in [0,1]$, and

\item $(\sum_{j=1}^{q-1} \kappa_{i-1}(Q_j))
+ \lambda \kappa_{i-1}(Q_q) > 0$.
\end{enumerate}

First observe that since
$(A,i)\not\in \mathcal{W}$ we have
\begin{align}
\sum_{j=1}^{q} g_{i-1}^{\mathcal{W}}(Q_j) &= 
 \sum_{j=1}^h g_{i-1}^{\mathcal{W}}(Q_j)
   = g_{i}^{\mathcal{W}}(A),
&& \text{(second equality follows by~\eqref{eq:gRec})}
\label{eq:gIWA} \\
\sum_{j=1}^{q} \kappa_{i-1}^{\mathcal{W}}(Q_j) &= 
 \sum_{j=1}^h \kappa_{i-1}^{\mathcal{W}}(Q_j)
   = \kappa_{i}^{\mathcal{W}}(A),
&& \text{(second equality follows by~\eqref{eq:kappaRec})}
\label{eq:kappaIWA}
\end{align}
where the first equality in the above statements
follows from 
$\kappa_{i-1}^\mathcal{W}(Q_j) = 0 =
 g_{i-1}^\mathcal{W}(Q_j)$ for
$j\in \{q+1,\dots,h\}$, since none of the 
sets $Q_{q+1},\dots, Q_h$ or any of its descendants
are contained in $\mathcal{W}$, by definition of
an efficient removal pattern.

Notice that $\kappa^{\mathcal{W}}_i(A)>0$ implies
by~\eqref{eq:kappaIWA} that
there is a $\bar{j}\in [q]$ such
that $0<\kappa_{i-1}^\mathcal{W}(Q_{\bar{j}})\leq
\kappa_{i-1}(Q_{\bar{j}})$,
and hence
$\rho_{i-1}(Q_{\bar{j}})<\infty$.
Because the auxiliary efficiencies $\rho_{i-1}(Q_j)$
are nonincreasing in $j$, we have $\rho_{i-1}(Q_q)<\infty$
which is equivalent to $\kappa_{i-1}(Q_q)>0$.
Hence, $\lambda$ is well defined and since
$\kappa_{i-1}(Q_q)\geq \kappa^{\mathcal{W}}_{i-1}(Q_q)$
we have $\lambda\in [0,1]$.
Furthermore,
\begin{align*}
0 &< \kappa_i^{\mathcal{W}}(A)
  = \sum_{j=1}^q \kappa_{i-1}^{\mathcal{W}}(Q_j)
  = \left( \sum_{j=1}^{q-1} \kappa_{i-1} (Q_j)\right)
       + \lambda \kappa_{i-1}(Q_q),
\end{align*}
where the first equality follows from~\eqref{eq:kappaIWA}.
We can thus apply Lemma~\ref{lem:prefixDom}
to the ratio in~\eqref{eq:gOverK1} to obtain
\begin{equation}\label{eq:appRatioLem}
\begin{aligned}
\frac{\sum_{j=1}^h g_{i-1}(Q_j)}%
  {\sum_{j=1}^h \kappa_{i-1}(Q_j)}
 &\leq
\frac{\left(\sum_{j=1}^{q-1} g_{i-1}(Q_j)\right)
 + \lambda g_{i-1}(Q_q)}%
  {\left(\sum_{j=1}^{q-1} \kappa_{i-1}(Q_j)\right)
 + \lambda \kappa_{i-1}(Q_q)} \\
 &= \frac{\left(\sum_{j=1}^{q-1}
 g^{\mathcal{W}}_{i-1}(Q_j)\right)
 + \lambda g_{i-1}(Q_q)}%
  {\sum_{j=1}^{q}
     \kappa^{\mathcal{W}}_{i-1}(Q_j)},
\end{aligned}
\end{equation}
where the equality follows by
the definition of $\lambda$ in the
denominator, and by using the observation that
$(Q_j, i-1)\in \mathcal{W}$ for
$j\in \{1,\dots, q-1\}$,
which implies
$g_{i-1}^{\mathcal{W}}(Q_j)=g_{i-1}(Q_j)$ and
$\kappa_{i-1}^{\mathcal{W}}(Q_j)=\kappa_{i-1}(Q_j)$.
We thus obtain
\begin{align*}
\frac{\kappa^{\mathcal{W}}_i(A)}{\kappa_i(A)}
\left( g_i(A) - 2^i \right)
 &\leq
   \frac{\kappa_i^{\mathcal{W}}(A)}%
     {\sum_{j=1}^h\kappa_{i-1}(Q_j)}
      \sum_{j=1}^h g_{i-1}(Q_j)
  && \text{(by~\eqref{eq:gOverK1})}
\\
 &\leq
    \frac{\sum_{j=1}^q
    \kappa_{i-1}^{\mathcal{W}}(Q_j)}%
     {\sum_{j=1}^h\kappa_{i-1}(Q_j)}
      \sum_{j=1}^h g_{i-1}(Q_j)
  && \text{(by~\eqref{eq:kappaIWA})}
\\
 &\leq
    \left( \sum_{j=1}^{q-1}
    g_{i-1}^{\mathcal{W}}(Q_j)\right)
   + \lambda g_{i-1}(Q_q).
  && \text{(by~\eqref{eq:appRatioLem})}
\end{align*}
Applying the induction hypothesis
to 
$\lambda(g_{i-1}(Q_q) - 2^{i-1}) =
  \frac{\kappa^{\mathcal{W}}_{i-1}(Q_q)}{\kappa_{i-1}(Q_q)}
    (g_{i-1}(Q_q)-2^{i-1})$
we get 
\begin{align*}
\lambda g_{i-1}(Q_q)
  &\leq g_{i-1}^{\mathcal{W}}(Q_q) + 2^{i-1}(1+\lambda)
   && \text{(induction hypothesis)}\\
  &\leq g_{i-1}^{\mathcal{W}}(Q_q) + 2^i,
    && \text{($\lambda\leq 1$)}
\end{align*}
and hence
\begin{align*}
\frac{\kappa_i^{\mathcal{W}}(A)}{\kappa_i(A)}
 (g_i(A) - 2^i) &\leq
  \left(\sum_{j=1}^q g_{i-1}^{\mathcal{W}}(Q_j)\right)
  + 2^i
  \\
 &=
   g_{i}^\mathcal{W}(A) + 2^i
,
   && \text{(by~\eqref{eq:gIWA})}
\end{align*}
thus proving the lemma.
\end{proof}
}

\begin{lemma}\label{lem:impactGGuarantee}
Let $\mathcal{W}$ be an efficient removal pattern
with corresponding removal set $R=R(\mathcal{W})$. Then
\begin{equation*}
g_p^{\mathcal{W}}(V)
\geq \frac{1}{2} \frac{c(R)}{c(U)} \val(U)
 - 2^{p}.
\end{equation*}
\end{lemma}
\ifbool{shortVersion}{}{%
\begin{proof}
The statement follows from
\begin{align*}
g_p^{\mathcal{W}}(V) &\geq
 \frac{\kappa_p^{\mathcal{W}}(V)}{\kappa_p(V)}
   \cdot (g_p(V)-2^p) - 2^p
 && \text{(by Lemma~\ref{lem:gToGW})}\\
 &\geq
   \frac{\kappa_p^{\mathcal{W}}(V)}{\kappa_p(V)}\val(U) - 2^p
 && \text{($g_p(V)-2^p \geq g_p(V)-2^{p+1} = \val(U)$
      by~\eqref{eq:gpV})}\\
 &= \frac{1}{2} \frac{\kappa^{\mathcal{W}}_p(V)}{c(U)}
    \val(U) - 2^p
 && \text{(by~\eqref{eq:kappapV})}\\
 &\geq \frac{1}{2} \frac{c(R)}{c(U)}\val(U) - 2^p
 && \text{(by~\eqref{eq:kappaWpV})}.
\end{align*}

\end{proof}
}

Combining Proposition~\ref{prop:impGW} and
Lemma~\ref{lem:impactGGuarantee} we obtain
the following.
\begin{corollary}\label{corr:impR}
Let $\mathcal{W}$ be an efficient removal pattern with
corresponding removal set $R=R(\mathcal{W})$.
Then
\begin{equation*}
\val(R) \geq \frac{1}{2}\frac{c(R)}{c(U)}\val(U)-3\cdot 2^{p-1}.
\end{equation*}
\end{corollary}

Now consider the interdiction set $R$ returned
by Algorithm~\ref{alg:mainAlg}.
If $c(R)=B$, Corollary~\ref{corr:impR} implies 
Theorem~\ref{thm:mainAlgThm}.
However, it may be that $c(R)$ only uses a very
small fraction of the available budget.
To prove Theorem~\ref{thm:mainAlgThm} we will show
how one can get around this problem by 
finding another efficient removal pattern $\mathcal{W}'$
that is over budget and whose value can be related
to $\val(R)$.
\ifbool{shortVersion}{Details on how this step can
be performed to finish the proof of Theorem~\ref{thm:mainAlgThm}
can be found in the long version of the paper.
}{
\begin{proof}[Proof of Theorem~\ref{thm:mainAlgThm}]
We will construct an efficient removal pattern
$\mathcal{W}'$ with corresponding removal set
$R'=R(\mathcal{W}')$ satisfying the following
two conditions:
\smallskip
\begin{compactenum}[(i)]
\item\label{item:RpB}
$c(R')\geq B$, and

\item\label{item:RtoRp}
$g^{\mathcal{W}}_p(V) \geq g^{\mathcal{W}'}_p(V) - 2^{p-1}$.
\end{compactenum}
\smallskip
First observe that the existence of $\mathcal{W}'$
indeed implies Theorem~\ref{thm:mainAlgThm} since
\begin{align*}
\val(R) &\geq
      g_{p}^{\mathcal{W}}(V) - 2^{p-1}
      && \text{(by Proposition~\ref{prop:impGW})}
  \\
     &\geq
       g_{p}^{\mathcal{W}'}(V) - 2^{p} 
       && \text{(using~\eqref{item:RtoRp})}
\\
     &\geq
   \frac{1}{2}\cdot \frac{c(R')}{c(U)}\val(U) - 2^{p+1}
       && \text{(by Lemma~\ref{lem:impactGGuarantee}
           applied to $\mathcal{W}'$)} \\
     &\geq
   \frac{1}{2}\cdot \frac{B}{c(U)}\val(U) - 2^{p+1}.
       && \text{(using~\eqref{item:RpB})}
\end{align*}

It remains to show that an 
efficient removal pattern
$\mathcal{W}'$ with the desired
properties~\eqref{item:RpB} and~\eqref{item:RtoRp}
exists.
We define $\mathcal{W}'$ in terms of $\mathcal{W}$.
Consider the construction of $\mathcal{W}$
through Algorithm~\ref{alg:mainAlg}.
Let $\ell\in \{-1,\dots, p-1\}$ be the last
iteration of the while loop of
Algorithm~\ref{alg:mainAlg} where the index $s$ was not
equal to $h$, i.e., the maximum possible value in that
iteration.
Hence, this corresponds to the lowest value of $\ell$
in which iteration we have $s\neq h$.
Note that there must have been an iteration where
$s\neq h$ since for otherwise $R=U$ which violates
the fact that $R$ is an interdiction set 
because $c(U)>B$.

Let $A\in \mathcal{A}_{\ell+1}$ be the set considered
by Algorithm~\ref{alg:mainAlg}
at the beginning of iteration $\ell$, and let
$\mathcal{C}_{\ell}(A)=\{Q_1,\dots, Q_h\}$ be the
numbering of the children of $A$ considered in that iteration.
Moreover, we denote by $\overline{\mathcal{W}}$
the set $\mathcal{W}$ at the beginning of iteration $\ell$.
We recall that $s$ is defined by 
\begin{equation*}
s =
  \max\left\{j\in \{0,\dots, h\}\;\Bigg\vert\;
     c\left(\overline{\mathcal{W}}\cup
         \{(Q_k, \ell) \mid k\in [j]\}
           \right) \leq B\right\}.
\end{equation*}
Let
\begin{equation*}
\mathcal{W}' = \overline{\mathcal{W}} \cup 
   \{(Q_k, \ell) \mid k\in [s+1]\}.
\end{equation*}
Clearly, $\mathcal{W}'$ is an efficient removal pattern.
Furthermore, the removal set $R'=R(\mathcal{W})$ satisfies
condition~\eqref{item:RpB}, i.e., $c(R')>B$, by definition
of $s$.
It remains to show that~\eqref{item:RtoRp} holds.

Notice that either $\ell=-1$, or all
children of $(Q_{s+1},\ell)$ are added to $\mathcal{W}$
as sets on level $\ell-1$, which follows from the
fact that 
$\ell$ was the last iteration of Algorithm~\ref{alg:mainAlg}
in which not all children were added to $\mathcal{W}$.
Moreover, $\mathcal{W}$ contains no sets on levels
$-1,\dots, \ell-2$: This clearly holds if $\ell=-1$;
otherwise, Algorithm~\ref{alg:mainAlg}
left the while loop after having added all children
of $Q_{s+1}$.
Hence, $\mathcal{W}$ and $\mathcal{W}'$ are almost identical
with the only difference that $\mathcal{W'}$ contains
$(Q_{s+1}, \ell)$, which is not contained in $\mathcal{W}$
and, if $\ell\neq -1$, then $\mathcal{W}$ contains
all children of $(Q_{s+1},\ell)$, which are not contained
in $\mathcal{W}'$.
This implies
\begin{align*}
g_p^{\mathcal{W}'}(V) = g_p^{\mathcal{W}}(V) + \max\{1,2^\ell\}.
\end{align*}
Point~\eqref{item:RtoRp} now follows by observing
that $\ell\leq p-1$ (and $p\geq 1$).

\end{proof}
}

\ifbool{shortVersion}{}{%
\section{An $O(1)$-approximation for metric TSP interdiction}
\label{sec:TSPInt}

We consider the metric TSP problem as highlighted in
Section~\ref{sec:intro}. We recall that
we are given an undirected connected
graph $G=(V,E)$ with edge lengths $\ell:E\rightarrow \mathbb{Z}_{>0}$
and the goal is to find a shortest closed walk that visits each vertex
at least once.
In its interdiction version, every edge is also given an
interdiction cost $c:E\rightarrow \mathbb{Z}_{>0}$, and there
is a global budget $B\in \mathbb{Z}_{> 0}$. The goal of
metric TSP interdiction is to find
a set $R\subseteq E$ of edges to interdict with $c(R)\leq B$, such
that the length of a shortest closed walk in $(V,E\setminus R)$
that visits each vertex at least once is as large as possible.

For any set $U\subseteq E$, we denote by $\TSP(U)$ the length
of a shortest closed walk in $(V,E\setminus U)$
visiting each vertex at least once.
To avoid trivial cases we assume that the graph cannot
be disconnected by removing
an interdiction set, i.e., for any 
$R\subseteq E$  with $c(R)\leq B$, the graph $(V,E\setminus R)$
is connected.
Formally, metric TSP interdiction can be described as follows:
\begin{equation}\label{eq:TSPInt}
\max \{ \TSP(R) \mid R\subseteq E, c(R)\leq B \}.
\end{equation}
The following result 
now easily follows by the fact that
$\MST(U)$ and $\TSP(U)$ are at most a factor of $2$ apart.

\begin{theorem}\label{thm:TSPtoMST}
Let $R\subseteq E$ be an interdiction set obtained by applying
an $\alpha$-approximation to the MST interdiction problem defined
on the graph $G$ with weights given by $\ell$, interdiction
costs given by $c$, and budget $B$.
Then $R$ is a $2\alpha$-approximation for metric TSP interdiction.
\end{theorem}
\begin{proof}
First observe that for any interdiction set $U\subseteq E$,
we have
\begin{equation}\label{eq:tspGeqMST}
\TSP(U) \geq \MST(U),
\end{equation}
because any solution to $\TSP(U)$ must connect all
vertices and therefore contains a spanning tree.
Furthermore, we also have for any interdiction set $U\subseteq E$,
\begin{equation}\label{eq:tspLeq2MST}
\TSP(U) \leq 2\MST(U),
\end{equation}
because doubling a spanning tree leads to a closed walk that
visits all vertices. This corresponds to the well-known
Double-Tree Algorithm which $2$-approximates metric
TSP (see~\cite{korte_2012_combinatorial}).
Let $R^*_{\MST}$ and $R^*_{\TSP}$ be optimal solutions
to the MST interdiction problem and the
metric TSP interdiction problem on $G$, respectively.
We thus obtain that our $\alpha$-approximation $R$ for
the MST interdiction problem satisfies
\begin{align*}
\TSP(R) &\geq \MST(R)
    && \text{(by~\eqref{eq:tspGeqMST})} \\
        &\geq \frac{1}{\alpha} \MST(R^*_{\MST})
    && \text{($R$ is an $\alpha$-approximation
          for MST interdiction)} \\
        &\geq \frac{1}{\alpha} \MST(R^*_{\TSP})
    && \text{($R^*_{\MST}$ is an optimal solution for MST interdiction)} \\
        &\geq \frac{1}{2\alpha} \TSP(R^*_{\TSP}).
    && \text{(by~\eqref{eq:tspLeq2MST})}
\end{align*}

\end{proof}

Finally, Theorem~\ref{thm:TSPmain} is a direct consequence of
Theorem~\ref{thm:TSPtoMST} and Theorem~\ref{thm:MSTmain},
our $14$-approximation for MST interdiction.

}

\section*{Acknowledgements}

We are grateful to Chandra Chekuri, R. Ravi, and 
the anonymous reviewers for many
helpful comments. 

\bibliographystyle{plain}

\ifbool{shortVersion}{}{
\appendix

\section{Relation to graph disconnection problems}\label{app:graphPart}

The $k$-cut problem is closely related
to MST interdiction through its budgeted version, the
maximum components problem (MCP).
We recall that MCP
asks to break a graph
$G=(V,E)$ into a maximum number of connected components by removing
a given number $q$ of edges.
The following
is a simple way to reduce MCP  
to an MST interdiction problem: Set $c(e)=1, w(e)=0\; \forall e\in E$,
set the budget $B=q$,
and add to $G$ a set of $|V|-1$ edges $T$ forming a spanning tree;
for $e\in T$ we set $w(e)=1$ and make sure that these edges
cannot be interdicted by setting $c(e)=B+1$.
One can easily check that this reduction preserves
objective values.
Another reduction that does not preserve the objective
values has been presented in~\cite{frederickson_1999_increasingC}.
A generalization of MCP, where
edges have interdiction costs, was considered
in~\cite{engelberg_2007_cut} and called the
\emph{budgeted graph disconnection} (BGD) problem.
These budgeted versions of the $k$-cut problem admit
$O(1)$-approximations by extending
ideas for $O(1)$-approximations
for $k$-cut~\cite{frederickson_1999_increasingC,engelberg_2007_cut}.

\section{Proof of Lemma~\ref{lem:prefixDom}}
\label{app:prefixDom}

We start by observing that we can assume $b_j>0$
for $j\in [k]$. Otherwise one can remove all pairs
$a_j,b_j$ with $b_j=0$ from the sequence. Doing so
leads to a sharper statement since the left-hand side
of the inequality claimed by the lemma
decreases at most as much as its right-hand side.
Hence, assume $b_j>0$ for $j\in [k]$.

For brevity we define $r_j= \frac{a_j}{b_j}$
for $j\in [k]$. 
If $q=k$ and $\lambda=1$, the statement trivially holds.
Hence, assume that either $q<k$ or $\lambda<1$.
We define the following expressions $\beta$ and $\gamma$,
where the denominator
of $\gamma$ must be strictly positive since
either $q<k$ or $\lambda<1$:
\begin{align*}
\beta  &= \frac{\left(\sum_{j=1}^{q-1} b_j r_j\right)
  + \lambda b_q r_q}%
 {\left( \sum_{j=1}^{q-1} b_j\right) + \lambda b_q},\\
\gamma &=
    \frac{(1-\lambda)b_q r_q + \sum_{j=q+1}^k b_jr_j}%
  {(1-\lambda)b_q + \sum_{j=q+1}^k b_j}.
\end{align*}
Notice that $\beta$ can be interpreted as a convex
combination of $r_1,\dots r_{q}$, and since
$r_1\geq \dots \geq r_{q}$, we have
$\beta\geq r_q$. Similarly, $\gamma$ is a convex
combination of $r_q,\dots, r_k$, and hence
$\gamma \leq r_q$. Thus, $\beta \geq \gamma$.
The result now follows by 
\begin{align}
\frac{\sum_{j=1}^k a_j}{\sum_{j=1}^k b_j}
 &= \frac{\sum_{j=1}^k b_j r_j}{\sum_{j=1}^k b_j}%
  \nonumber \\
 &= \frac{1}{\sum_{j=1}^k b_j} \left[
    \left( \left(\sum_{j=1}^{q-1} b_j\right)
     + \lambda b_q \right)
    \beta
   + \left((1-\lambda)b_q + \sum_{j=q+1}^k b_j\right)
    \gamma
  \right]\nonumber\\
 &\leq \beta = \frac{\left(\sum_{j=1}^{q-1} a_j\right)
     + \lambda a_q}%
       {\left(\sum_{j=1}^{q-1} b_j \right)
     + \lambda b_q} \nonumber,
\end{align}
where the inequality follows by upper bounding
$\gamma$ by $\beta$.

\section{Details on erroneous claim in~\cite{shen_1999_finding}}
\label{app:shen}

The article~\cite{shen_1999_finding} presents several algorithms
for the $k$ most vital edges problem for MST. In particular,
they claim to present a $2$-approximation.
However, their results are based on an erroneous
claim about spanning trees, which is stated as
Lemma~2 in~\cite{shen_1999_finding}.
In this section, after introducing
some basic notions used in~\cite{shen_1999_finding},
we state Lemma~2 of~\cite{shen_1999_finding} and
provide a counterexample for
it. Furthermore, we give a brief explanation of why the proof
of Lemma~2 that is presented in~\cite{shen_1999_finding}
is erroneous.

Let $G=(V,E)$ be an undirected graph with edge weights
$w:E\rightarrow \mathbb{Z}_{\geq 0}$,
and let $k\in \mathbb{Z}_{>0}$.
All edge weights are assumed to be distinct, and hence,
the MST is unique, also in any connected subgraph of $G$.
Furthermore, we assume that $G$ is $(k+1)$-edge-connected
to avoid the trivial case that the graph can be disconnected.
Let $T\subseteq E$ be the unique MST in $G$.
For each $e\in T$, let
\begin{equation*}
R(e) = \{f\in E \mid (T\cup\{f\})\setminus \{e\} \text{ is
a spanning tree}\}.
\end{equation*}
In~\cite{shen_1999_finding}, the edges in $R(e)$ are
called \emph{replacement edges for $e$}
since they can replace
$e$ in $T$ to obtain again a spanning tree.
Furthermore $R_e\subseteq R(e)$ is the set containing
the $k$ lightest edges in $R(e)$, i.e., these are
the $k$ lightest replacement edges for $e$.
Moreover, let $R=\cup_{e\in T} R_e$.
We are now ready to state the erroneous lemma
in~\cite{shen_1999_finding}.
\begin{lem2shen*}
Let $K$ be an optimal solution for the $k$ most vital
edges problem for MST. Then
\begin{equation*}
K\subseteq  T \cup R.
\end{equation*}
\end{lem2shen*}

The weighted graph depicted in Figure~\ref{fig:counterShen}
is a counterexample to the above Lemma.
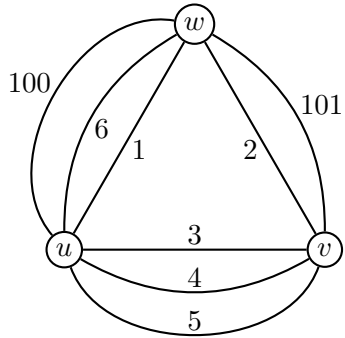
\begin{figure}[H]
\begin{center}
\begin{tikzpicture}

\begin{scope}[every node/.style={thick,draw=black,fill=white,circle,minimum size=13, inner sep=2pt}]
\def\rad{2cm}
\node (u) at (210:\rad) {$u$};
\node (v) at (-30:\rad) {$v$};
\node (w) at (90:\rad) {$w$};
\end{scope}

\begin{scope}[thick]
\draw (u) to[bend left=0pt] node[above=-2pt] {$3$} (v);
\draw (u) to[bend left=-30pt] node[above=-2pt] {$4$} (v);
\draw (u) to[bend left=-70pt] node[above=-2pt] {$5$} (v);

\draw (v) to[bend left=0pt] node[below left=-3pt] {$2$} (w);
\draw (v) to[bend left=-30pt] node[above right=-3pt] {$101$} (w);

\draw (w) to[bend left=0pt] node[below right=-3pt] {$1$} (u);
\draw (w) to[bend left=-30pt] node[below right=-3pt] {$6$} (u);
\draw (w) to[bend left=-70pt] node[above left=-3pt] {$100$} (u);

\end{scope}

\end{tikzpicture}

\end{center}
\caption{A counterexample to Lemma~2
in~\cite{shen_1999_finding} for $k=3$.}
\label{fig:counterShen}
\end{figure}
Its minimum spanning tree consists of the edges
of weight $1$ and $2$. For each of these edges,
the three best replacement edges are the edges
of weight $3$, $4$, and $5$.
No matter which three edges are removed among
the edges of weight $1$, $2$, $3$, $4$, and $5$, there
is always a spanning tree left that uses neither of
the two edges of weight $100$ and $101$, respectively.
However, removing the edges of weight $1$,$2$, and $6$,
leads to a graph whose minimum spanning tree contains
the edge of weight $100$.

Notice that the example in Figure~\ref{fig:counterShen} can easily
be converted to a simple graph (i.e., without parallel edges).
For example, this can be done by replacing each of the three
vertices by a clique of size $5$, where all edges in the clique
have very low weight and thus are not worth being removed; because
no matter which $3$ edges get removed, the vertices of any clique
can still be connected by low weight edges within the clique.
Each remaining edge connects the two cliques that correspond to
its endpoints, where it does not matter to which particular
vertex of a clique an edge is connected to, as long as no
parallel edges are created. Clearly, the edges
can be placed in a way to obtain a simple graph.

The main mistake in the proof of Lemma~2 presented
in~\cite{shen_1999_finding} is the assumption that
for any subset $U\subseteq T$,
one can simultaneously replace in $T$ each edge $e\in U$ by
an edge in $R(e)$, still obtaining a spanning tree.

}

\end{document}